\documentclass[reprint,longbibliography,a4paper,twocolumn]{revtex4-2}
\pdfoutput=1

\usepackage{amsmath,amssymb,amsthm}
\usepackage{graphics}
\usepackage{enumitem}
\usepackage{mathtools}
\usepackage{array}
\usepackage[hidelinks]{hyperref}

\setlist[description]{leftmargin=0em,listparindent=\parindent,parsep=0pt,itemsep=8pt plus 4pt minus 2pt}

\theoremstyle{plain}
\newtheorem*{theorem}{Theorem}

\DeclareMathOperator{\tr}{tr}
\DeclareMathOperator{\rank}{rank}
\DeclareMathOperator{\dimension}{dim}
\DeclareMathOperator*{\esssup}{ess\,sup}
\DeclareMathOperator*{\essinf}{ess\,inf}

\newcommand{\adja}[1]{{#1}}
\newcommand{\adjb}[1]{\breve{#1}}

\newcolumntype{C}[1]{>{\centering\arraybackslash}m{#1}}
\newcolumntype{L}[1]{>{\raggedright\arraybackslash}m{#1}}
\newcolumntype{R}[1]{>{\raggedleft\arraybackslash}m{#1}}
\newcolumntype{J}[1]{>{\arraybackslash}m{#1}}

\begin{document}

\title{Model Risk Analysis for Entropic Hedging Strategies}
\author{Paul McCloud}
\affiliation{McCloud Research}
\date{2 November 2026}

\begin{abstract}
Market incompleteness causes residual profit and loss after hedging that is accounted for with valuation adjustments and bid-offer spreads acting as model risk buffers for the derivative price. Entropic risk optimisation is a general framework for modelling these adjustments. When the prices are quadratic Gaussian, the coupled equations for price and hedge ratios are solved exactly, with a decomposition of the entropic risk into convexity, dimensionality, funding and strategy contributions that are attributed in the explanation of realised P\&L.

The results are applicable when the underlying prices follow Gaussian processes, such as fractional Ornstein-Uhlenbeck processes. The article concludes with analysis of the continuous hedging limit for these processes, and a discussion of their application in data-driven regressions.
\end{abstract}

\maketitle

When the underlying market is incomplete for the derivative risks, residual profit and loss generated by the unhedged factors must be mitigated with a bid-offer spread around the mid valuation acting as a model risk buffer. The risk decomposes into four contributions.
\begin{description}
\item[Convexity]The convexity between the derivative and underlying prices cannot be cancelled by linear replication. This contribution reduces as the hedge frequency is increased \cite{Boyle1980,Bertsimas2000}, at the expense of rebalancing costs \cite{Leland1985,Hodges1989}, but a baseline risk persists due to the possibility of market dislocations.
\item[Dimensionality]Risk factors that are material for the derivative but inconsequential for the underlying market cannot be hedged. These additional dimensions persist as residual risks for the derivative.
\item[Funding]Even if it is possible to offset derivative risks in the underlying market, funding the hedged position reintroduces risks when the derivative premium does not exactly match its hedge cost. Making the unrealistic assumption of risk-free funding masks these additional risks and the valuation adjustments they incur \cite{Piterbarg2010,Burgard2011,McCloud2013a,McCloud2013b}.
\item[Strategy]If the investor intentionally funds the derivative suboptimally, the unhedged risk factors contribute to the P\&L. This is the only contribution that the investor exercises direct control over, and is eliminated when the hedging strategy is optimised.
\end{description}
This article contributes the exact price and hedge from entropic risk optimisation \cite{McCloud2025a,McCloud2025b} for quadratic Gaussian prices \cite{Ahn2002,Leippold2002,McCloud2010,McCloud2012}, a decomposition of the entropic risk into four penalties that are matched in each case by an attribution in realised P\&L, and analysis of their scale behaviour in the continuous hedging limit for Gaussian processes. Exact results are derived for discrete fractional Ornstein-Uhlenbeck processes of the kind considered in \cite{Cheridito2003,Gatheral2018}, and the expressions provide second-order approximations for more general models in a form that could be used in data-driven regressions \cite{Buehler2019,Buehler2021,Murray2022}.

Entropic risk optimisation fits within a long tradition of entropic methods in mathematical finance applied to incomplete markets. The ambiguity in the choice of equivalent risk-neutral measure \cite{Harrison1981} can be resolved by minimising entropic distance from the economic measure \cite{Stutzer1996,Frittelli2000}, a solution that also emerges from the optimisation of exponential utility \cite{Rouge2000,Delbaen2002}. In the context of convex risk metrics \cite{Follmer2002,Frittelli2002}, entropic risk becomes a target for hedge optimisation \cite{Barrieu2005}. Combined with indifference pricing \cite{Musiela2004}, the theory generates price and hedge ratios with bid-offer corrections that mitigate the model risk from residual unhedged P\&L, with penalties that reduce to variances in the small notional limit, similar to the case of quadratic hedging \cite{Follmer1991,Duffie1991,Schweizer1995}.

A key advantage of entropic risk is its simplicity and tractability when applied to normal distributions. With quadratic Gaussian prices, the case explored in depth in this article, entropic risk is evaluated on quadratic forms in normal variables \cite{Imhof1961,Das2025} with results that require only elementary matrix algebra. All the contributions to risk that impact pricing and hedging are resolved with a clean decomposition that is directly attributable in realised P\&L. The decomposition aids trading intuition and the allocation of P\&L between originating, hedging and funding desks.

\section*{Quadratic Gaussian theorem}

\begin{figure*}[!t]
\centering
\includegraphics[width=0.4\textwidth]{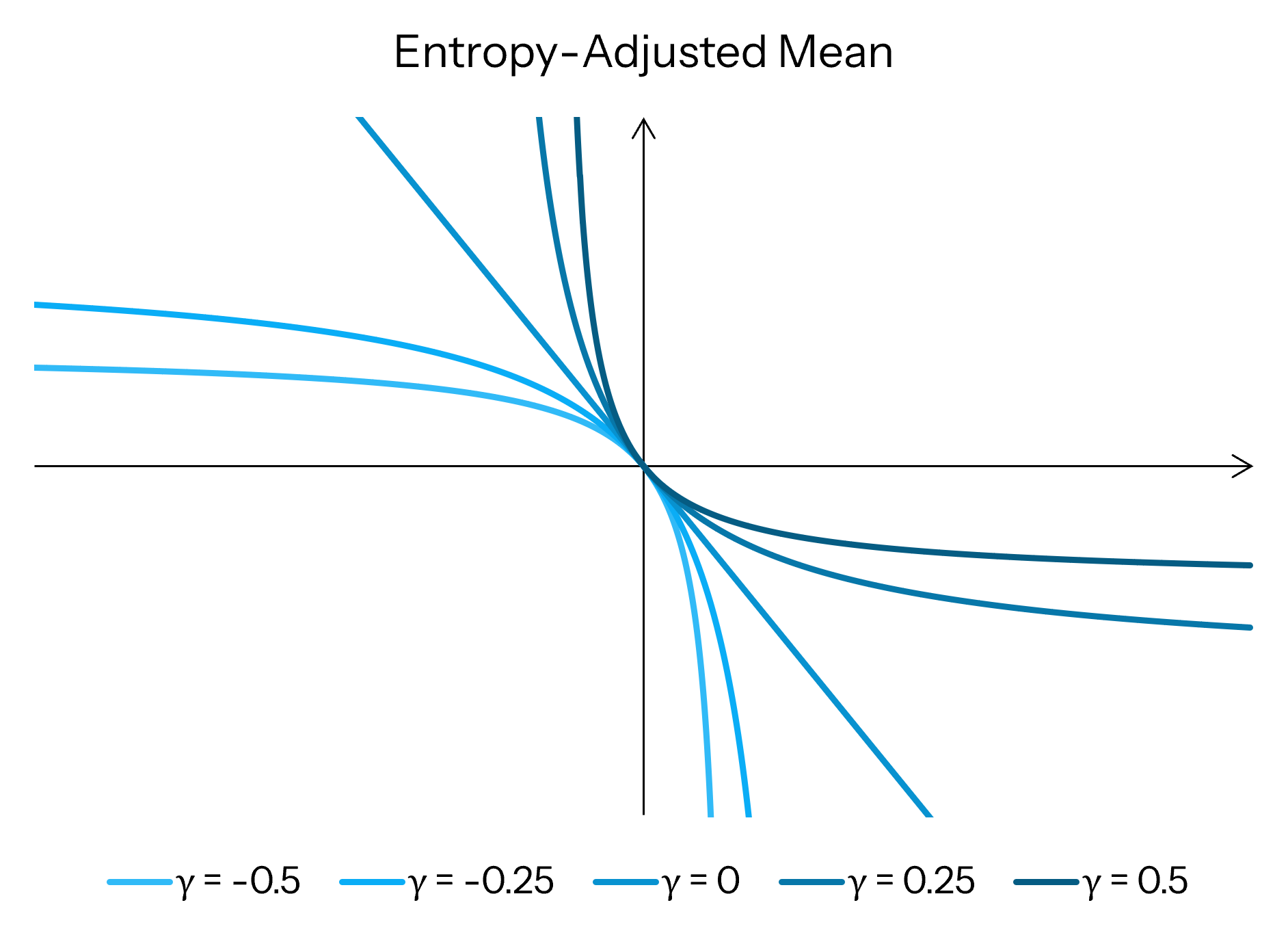}
\caption{Entropy-adjusted mean as a function of risk aversion $\alpha$ for one-factor quadratic Gaussian variables with a range of quadratic coefficients $\gamma$. The variables are standardised by linear transformation to have zero mean and unit standard deviation, and the central line is for the standard normal variable with $\gamma=0$. For $\gamma\ne0$, the curve is defined on a half line containing $\alpha=0$ which is bounded below when $\gamma>0$ and bounded above when $\gamma<0$.}
\end{figure*}

Entropic risk optimisation uses the entropy-adjusted mean defined by:
\begin{equation}
\mathbb{E}[\alpha][X]:=-\frac{1}{\alpha}\log\mathbb{E}\exp[-\alpha X]
\end{equation}
for a random variable $X$, where $\mathbb{E}$ is the expectation operator and $\alpha\ne0$ is the risk aversion. The definition is extended continuously to the case $\alpha=0$ with the definition $\mathbb{E}[0][X]:=\mathbb{E}[X]$. It is decreasing as a function of $\alpha$ and spans the spectrum of $X$:
\begin{align}
\esssup{X}=\qquad\qquad\qquad& \\
\lim_{\alpha\to-\infty}\mathbb{E}[\alpha][X]>\mathbb{E}[\alpha]&[X]>\lim_{\alpha\to\infty}\mathbb{E}[\alpha][X] \notag \\
&\qquad\qquad\qquad=\essinf{X} \notag
\end{align}
With this property, the entropy-adjusted mean is an effective measure of risk-adjusted expectation that can be tuned to a target risk aversion. Entropic risk, defined as the negative of the entropy-adjusted mean, measures the downside ($\alpha>0$) or upside ($\alpha<0$) risk of the variable.

Applied to quadratic Gaussian variables, the entropy-adjusted mean can be evaluated in closed form, which reduces later development to matrix algebra. The essential result is the following.
\begin{theorem}[Quadratic Gaussian]
Let $Z$ be a vector of independent standard normal variables in the measure $\mathbb{E}$, and define the quadratic Gaussian variable:
\begin{equation}
X:=\eta+\epsilon\cdot Z+\frac{1}{2}Z\cdot\gamma Z
\end{equation}
with constant coefficient scalar $\eta$, linear coefficient vector $\epsilon$ and quadratic coefficient matrix $\gamma$ satisfying the symmetry property $\gamma^\dagger=\gamma$.

Using the spectral decomposition of the quadratic coefficient, the variable $X$ is identified as the weighted sum of independent standard normal and noncentral chi-squared variables. The standard normal contribution originates from the zero eigenvalue, and the noncentral chi-squared contributions originate from the non-zero eigenvalues. Since the zero eigenvalue needs to be handled separately, this is made explicit in the spectral decomposition:
\begin{align}
1&=1_0+\sum\nolimits_{v\ne0}1_v \\
\gamma&=\sum\nolimits_{v\ne0}v1_v \notag
\end{align}
in terms of the orthogonal projections $1_v$ onto the ei\-gen\-spaces for the distinct eigenvalues $v$. The variable $X$ is written as:
\begin{equation}
X=\eta-\sum\nolimits_{v\ne0}\frac{\|1_v\epsilon\|^2}{2v}+\|1_0\epsilon\|W+\frac{1}{2}\sum\nolimits_{v\ne0}v C_v
\end{equation}
This expression decomposes the variable into a linear combination of independent variables, including the standard normal variable:
\begin{equation}
W:=\frac{1_0\epsilon\cdot1_0Z}{\|1_0\epsilon\|}
\end{equation}
when $\|1_0\epsilon\|\ne0$, and the noncentral chi-squared variables:
\begin{equation}
C_v:=\|1_v(Z+\epsilon/v)\|^2
\end{equation}
with noncentrality parameter $\|1_v\epsilon/v\|^2$ and degrees of freedom $\rank[1_v]$ for each non-zero eigenvalue $v$.

When the eigenvalues of the quadratic coefficient are greater than minus one, the variable $X$ satisfies the Girsanov relationship:
\begin{align}
\mathbb{E}[\exp[-X]f[Z]]={}&\frac{\exp[-\eta+\frac{1}{2}\epsilon\cdot(1+\gamma)^{-1}\epsilon]}{\sqrt{\det[1+\gamma]}}\times \\
&\mathbb{E}[f[(1+\gamma)^{-1/2}Z-(1+\gamma)^{-1}\epsilon]] \notag
\end{align}
obtained by completing the square in the exponent of the Gaussian integrand. Introduce the risk aversion $\alpha$ and apply the Girsanov relationship for the variable $\alpha X$ to obtain the entropy-adjusted mean:
\begin{align}
\mathbb{E}[\alpha]&[X]=\eta+\frac{1}{2}\tr[\gamma] \\
&-\frac{1}{2}\alpha\left(\epsilon\cdot(1+\alpha\gamma)^{-1}\epsilon+\tr\!\left[\frac{\alpha\gamma-\log[1+\alpha\gamma]}{\alpha^2}\right]\right) \notag
\end{align}
making use of the spectral decomposition to define the function of $\gamma$. The first two terms sum to the mean of $X$, and the two terms in the brackets are both non-negative and sum to the variance of $X$ in the limit $\alpha\to0$. The Girsanov relationship also generates the chain rule for the entropy-adjusted mean:
\begin{align}
\mathbb{E}&[\alpha][X+f[Z]]= \\
&\mathbb{E}[\alpha][X]+\mathbb{E}[\alpha][f[(1+\alpha\gamma)^{-1/2}Z-\alpha(1+\alpha\gamma)^{-1}\epsilon]] \notag
\end{align}
These expressions are valid for $\alpha$ within an open interval containing zero. Let $\gamma_-$ and $\gamma_+$ be the lowest and highest eigenvalues of $\gamma$. If $\gamma_-<0$, the interval has upper bound $\alpha=-1/\gamma_->0$, otherwise it is unbounded above. If $\gamma_+>0$, the interval has lower bound $\alpha=-1/\gamma_+<0$, otherwise it is unbounded below.
\end{theorem}

\begin{proof}
The decomposition of $X$ into standard normal and noncentral chi-squared variables follows from the spectral decomposition of $\gamma$ as provided. Use this decomposition together with the orthogonality $1_u1_v=\delta_{uv}1_v$ and symmetry $1_v^\dagger=1_v$ of the projections to generate the decomposition of $X$:
\begin{align}
X={}&\eta+1_0\epsilon\cdot 1_0Z \\
&+\sum\nolimits_{v\ne0}(1_v\epsilon\cdot 1_vZ+\frac{1}{2}v1_vZ\cdot 1_vZ) \notag \\
={}&\eta-\sum\nolimits_{v\ne0}\frac{\|1_v\epsilon\|^2}{2v}+\|1_0\epsilon\|\frac{1_0\epsilon\cdot 1_0Z}{\|1_0\epsilon\|} \notag \\
&+\frac{1}{2}\sum\nolimits_{v\ne0}v\|1_v(Z+\epsilon/v)\|^2 \notag
\end{align}

The expectation is given by the ratio of Gaussian integrals:
\begin{align}
\mathbb{E}[\exp[-X]&f[Z]]= \\
&\frac{\int_z\exp[-\eta-\epsilon\cdot z-\frac{1}{2}z\cdot(1+\gamma)z]f[z]\,dz}{\int_z\exp[-\frac{1}{2}z\cdot z]\,dz} \notag
\end{align}
Complete the square in the exponent via the change of variables:
\begin{align}
y&=(1+\gamma)^{1/2}z+(1+\gamma)^{-1/2}\epsilon \\
z&=(1+\gamma)^{-1/2}y-(1+\gamma)^{-1}\epsilon \notag
\end{align}
Use this transformation to generate the expressions:
\begin{align}
dz&=\frac{1}{\sqrt{\det[1+\gamma]}}\,dy \\
-\epsilon\cdot z-\frac{1}{2}z\cdot(1+\gamma)z&=\frac{1}{2}\epsilon\cdot(1+\gamma)^{-1}\epsilon-\frac{1}{2}y\cdot y \notag
\end{align}
These expressions are inserted into the integrals to complete the proof. The results for the entropy-adjusted mean follow directly from the Girsanov relationship by replacing $X$ with $\alpha X$ and $f[Z]$ with $\exp[-\alpha f[Z]]$.
\end{proof}
The Girsanov relationship is the integrated form of the familiar drift adjustment result for equivalent measures generated by Brownian motion. When $X$ is quadratic in $Z$, the measure change can be implemented as a linear transformation to the distribution of the factors. Crucially, this transformation converts quadratic Gaussian variables to quadratic Gaussian variables. In the following, the decomposition of risk and P\&L is generated from a chain of four consecutive measure changes, each of which preserves the quadratic Gaussian nature of the variables, with the linear transformation of the Girsanov relationship applied at each step.

\section*{Quadratic Gaussian pricing}

Let $q$ and $Q$ be the initial and final price vectors for the underlying securities, and let $p$ and $P$ be the initial and final price scalars for the derivative security. The final prices are modelled as quadratic Gaussian variables:
\begin{align}
Q&:=\adja{\mu}+\adja{\sigma}Z \\
P&:=\adja{\eta}+\adja{\epsilon}\cdot Z+\frac{1}{2}Z\cdot\adja{\gamma}Z \notag
\end{align}
where $Z$ is a vector of independent standard normal variables, and the number $N_f:=\dimension[Z]$ of factors is not assumed to be equal to the number $N_m:=\dimension[Q]$ of underlyings. For the underlyings, the constant coefficient $\adja{\mu}$ has dimensions $(N_m,1)$ and the linear coefficient $\adja{\sigma}$ has dimensions $(N_m,N_f)$. For the derivative, the constant coefficient $\adja{\eta}$ has dimensions $(1,1)$, the linear coefficient $\adja{\epsilon}$ has dimensions $(N_f,1)$ and the quadratic coefficient $\adja{\gamma}$ has dimensions $(N_f,N_f)$ and is symmetric, $\adja{\gamma}^\dagger=\adja{\gamma}$. No additional funding securities are assumed, and a portfolio with notionals $\omega$ of the underlyings and notional $\lambda$ of the derivative is self-funded when it satisfies $\omega\cdot q+\lambda p=0$.

Three additional assumptions are made on the underlying market.
\begin{description}
\item[Fundability]The initial underlying prices satisfy the condition $q\ne0$. This ensures that at least one underlying is available to fund the derivative.
\item[Independence]The covariance $\adja{\sigma}\adja{\sigma}^\dagger$ of the final underlying prices is positive definite. This ensures that the underlyings hedge distinct risk factors, and eliminates the possibility of arbitrage in the underlying market.
\item[Discounting]The underlying prices satisfy the condition $q\cdot(\adja{\sigma}\adja{\sigma}^\dagger)^{-1}\adja{\mu}>0$. This means that the initial prices and the mean of the final prices do not point away from each other in the inner product generated by $(\adja{\sigma}\adja{\sigma}^\dagger)^{-1}$, a condition that ensures discount factors are positive.
\end{description}

The assumed forms for the underlying and derivative prices are limiting in scope, and the approach adopted here is to derive exact solutions for an approximate model, rather than derive approximate solutions for an exact model. They are, nonetheless, sufficient to express the key features of market incompleteness, and extend beyond their immediate domain of validity as regression models. The underlying prices are in the form that a Gaussian process generates over the hedging interval, and the derivative price adds quadratic terms that arise in the second-order approximation of a more general expression. The model exhibits convexity between underlying and derivative prices, the factor dimensions that are not spanned by the underlyings constitute unhedgeable risks for the derivative, and the absence of a risk-free underlying creates funding risks.

For a portfolio comprising notionals $\omega$ of the underlyings and notional $\lambda$ of the derivative, the final price of the portfolio is the quadratic Gaussian variable:
\begin{align}
\omega\cdot Q&+\lambda P= \\
&(\lambda\adja{\eta}+\omega\cdot\adja{\mu})+(\lambda\adja{\epsilon}+\adja{\sigma}^\dagger\omega)\cdot Z+\frac{1}{2}Z\cdot\lambda\adja{\gamma}Z \notag
\end{align}
Using the quadratic Gaussian theorem, the entropy-adjusted mean of the final price is:
\begin{align}
\mathbb{E}[\alpha]&[\omega\cdot Q+\lambda P]=(\lambda\adja{\eta}+\omega\cdot\adja{\mu})+\frac{1}{2}\tr[\lambda\adja{\gamma}] \\
&-\frac{1}{2}\alpha\left((\lambda\adja{\epsilon}+\adja{\sigma}^\dagger\omega)\cdot(1+\alpha\lambda\adja{\gamma})^{-1}(\lambda\adja{\epsilon}+\adja{\sigma}^\dagger\omega)\vphantom{\frac{1}{1}}\right. \notag \\
&\left.\qquad\qquad+\tr\!\left[\frac{\alpha\lambda\adja{\gamma}-\log[1+\alpha\lambda\adja{\gamma}]}{\alpha^2}\right]\right) \notag
\end{align}
which is valid for $\alpha\lambda$ within an open interval containing zero.

Entropic risk optimisation maximises the entropy-adjusted mean, at the risk aversion $\alpha>0$ of the investor, for the final price of a self-funded portfolio. The program of entropic pricing \cite{McCloud2025a} follows three steps.
\begin{description}
\item[Calibration]Portfolio optimisation in the underlying market finds the performance target for the strategy.
\item[Hedging]Adding the derivative, the hedge is the incremental underlying portfolio that restores optimality.
\item[Pricing]Indifference dictates that the performance of the derivative portfolio matches the target.
\end{description}
The quadratic Gaussian theorem yields exact expressions for the optimal and hedge portfolios and indifference price in the quadratic Gaussian model.

\subsection*{Calibration}

Define the funding rate $r$ and optimal portfolio $\phi$ by:
\begin{align}
r&:=\frac{q\cdot(\adja{\sigma}\adja{\sigma}^\dagger)^{-1}(\adja{\mu}-q)}{q\cdot(\adja{\sigma}\adja{\sigma}^\dagger)^{-1}q} \\
\phi&:=(\adja{\sigma}\adja{\sigma}^\dagger)^{-1}(\adja{\mu}-(1+r)q) \notag
\end{align}
These expressions are well defined thanks to the fundability and independence assumptions, and the funding rate satisfies $1+r>0$ thanks to the discounting assumption. The optimal portfolio is a sum of the legs $(\adja{\sigma}\adja{\sigma}^\dagger)^{-1}\adja{\mu}$ and $(\adja{\sigma}\adja{\sigma}^\dagger)^{-1}q$ weighted so that it is self-funded, $\phi\cdot q=0$.

For an underlying portfolio with notionals $\omega$, the entropy-adjusted mean at risk aversion $\alpha>0$ of the final price is expressed in terms of the funding rate and optimal portfolio:
\begin{align}
\mathbb{E}[\alpha][\omega\cdot Q]={}&(1+r)\omega\cdot q \\
&+\frac{1}{2\alpha}(\|\adja{\sigma}^\dagger\phi\|^2-\|\adja{\sigma}^\dagger(\phi-\alpha\omega)\|^2) \notag
\end{align}
Optimising across all portfolios that satisfy the self-funding condition $\omega\cdot q=0$, the entropy-adjusted mean is maximised at the portfolio $\omega^\ast:=\phi/\alpha$ proportional to $\phi$, with maximum entropy-adjusted mean:
\begin{equation}
\mathbb{E}[\alpha][\omega^\ast\cdot Q]=\frac{1}{2\alpha}\|\adja{\sigma}^\dagger\phi\|^2
\end{equation}

\subsection*{Hedging}

In the next step of the program, consider the derivative funded by the underlying portfolio with notionals $\delta$ satisfying the self-funding condition $\delta\cdot q=p$. Add notional $\lambda$ of the funded derivative to the optimal underlying portfolio, so that the combined portfolio has initial price $0$ and final price $\omega^\ast\cdot Q+\lambda(P-\delta\cdot Q)$. The quadratic Gaussian theorem generates the increment in the entropy-adjusted mean resulting from the introduction of the derivative. Define the coefficients adjusted for measure and convexity:
\begin{align}
\adjb{\sigma}&:=\adja{\sigma}(1+\alpha\lambda\adja{\gamma})^{-1/2} \\
\adjb{\eta}&:=\adja{\eta}-\adja{\epsilon}\cdot\adja{\sigma}^\dagger\phi+\frac{1}{2}\adja{\sigma}^\dagger\phi\cdot\adja{\gamma}\adja{\sigma}^\dagger\phi+\frac{1}{2}\tr[\adja{\gamma}] \notag \\
\adjb{\epsilon}&:=(1+\alpha\lambda\adja{\gamma})^{-1/2}(\adja{\epsilon}-\adja{\gamma}\adja{\sigma}^\dagger\phi) \notag
\end{align}
and the projection onto the subspace spanned by $\adjb{\sigma}$:
\begin{equation}
1_{\adjb{\sigma}}:=\adjb{\sigma}^\dagger(\adjb{\sigma}\adjb{\sigma}^\dagger)^{-1}\adjb{\sigma}
\end{equation}
noting that the invertibility of $\adjb{\sigma}\adjb{\sigma}^\dagger$ follows from the invertibility of $\adja{\sigma}\adja{\sigma}^\dagger$. The hedging leg $\delta^\sigma$, funding leg $\delta^s$ and funding spread $s$ are then defined by:
\begin{align}
\delta^\sigma&:=(\adjb{\sigma}\adjb{\sigma}^\dagger)^{-1}\adjb{\sigma}\adjb{\epsilon} \\
\delta^s&:=(\adjb{\sigma}\adjb{\sigma}^\dagger)^{-1}q \notag \\
s&:=\frac{p-\delta^\sigma\cdot q}{\delta^s\cdot q} \notag
\end{align}
and the hedge portfolio $\delta^\ast:=\delta^\sigma+s\delta^s$ is expressed as a sum of the hedging and funding legs weighted so that it satisfies the self-funding condition $\delta^\ast\cdot q=p$.

\begin{figure*}[!t]
\begin{alignat*}{5}
& &\qquad\qquad &\makebox[0pt][l]{Penalty} && &\qquad\qquad &\makebox[0pt][l]{Attribution} && \\
&\text{Convexity:} & &\chi^\gamma & &=\tr\!\left[\frac{\alpha\lambda\adja{\gamma}-\log[1+\alpha\lambda\adja{\gamma}]}{(\alpha\lambda)^2}\right] & &X^\gamma & &=\frac{1}{2}((Z+\adja{\sigma}^\dagger\phi)\cdot\adja{\gamma}(Z+\adja{\sigma}^\dagger\phi)-\tr[\adja{\gamma}]) \notag \\
&\text{Dimensionality:} & &\chi^\sigma & &=\|(1-1_{\adjb{\sigma}})\adjb{\epsilon}\|^2 & &X^\sigma & &=(\adja{\epsilon}-\adja{\gamma}\adja{\sigma}^\dagger\phi)\cdot(Z+\adja{\sigma}^\dagger\phi)-\delta^\sigma\cdot(Q-(1+r)q) \notag \\
&\text{Funding:} & &\chi^s & &=s^2\delta^s\cdot q & &X^s & &=-s\delta^s\cdot(Q-(1+r)q) \notag \\
&\text{Strategy:} & &\chi^\delta & &=\|\adjb{\sigma}^\dagger(\delta-\delta^\ast)\|^2 & &X^\delta & &=-(\delta-\delta^\ast)\cdot Q \notag
\end{alignat*}
\caption{The P\&L of the funded derivative is initially assessed by its risk penalties and finally explained by its attribution. These decompose into convexity, dimensionality, funding and strategy contributions, where the penalty is generated by the incremental entropy-adjusted mean from the corresponding attribution. The investor eliminates the strategy penalty via optimal hedging. The remaining penalties, which are all non-negative, are the result of market incompleteness and cannot be eliminated.}
\end{figure*}

The central result of this article is that the P\&L and entropic risk of the funded derivative can be decomposed into four contributions from convexity, dimensionality, funding and strategy:
\begin{align}
&P-\delta\cdot Q= \\
&\qquad(\adjb{\eta}-(1+r)p)+(X^\gamma+X^\sigma+X^s+X^\delta) \notag \\
&\frac{1}{\lambda}(\mathbb{E}[\alpha][\omega^\ast\cdot Q+\lambda(P-\delta\cdot Q)]-\mathbb{E}[\alpha][\omega^\ast\cdot Q])= \notag \\
&\qquad(\adjb{\eta}-(1+r)p)-\frac{1}{2}\alpha\lambda(\chi^\gamma+\chi^\sigma+\chi^s+\chi^\delta) \notag
\end{align}
The contributions to the entropic risk penalise the performance of the portfolio, as the penalties $(\chi^\gamma,\chi^\sigma,\chi^s,\chi^\delta)$ are all non-negative, and each represents an unhedged risk assessed at the start of the hedging interval. The realised contributions of these risks to the P\&L are then attributed at the end of the hedging interval with the attributions $(X^\gamma,X^\sigma,X^s,X^\delta)$.
\begin{description}
\item[Convexity]Define the penalty $\chi^\gamma$ and attribution $X^\gamma$ for convexity:
\begin{align}
\chi^\gamma&:=\tr\!\left[\frac{\alpha\lambda\adja{\gamma}-\log[1+\alpha\lambda\adja{\gamma}]}{(\alpha\lambda)^2}\right] \\
X^\gamma&:=\frac{1}{2}((Z+\adja{\sigma}^\dagger\phi)\cdot\adja{\gamma}(Z+\adja{\sigma}^\dagger\phi)-\tr[\adja{\gamma}]) \notag
\end{align}
The chain rule in the quadratic Gaussian theorem generates the convexity contribution to the entropy-adjusted mean:
\begin{align}
\frac{1}{\lambda}&(\mathbb{E}[\alpha][\omega^\ast\cdot Q+\lambda X^\gamma] \\
&-\mathbb{E}[\alpha][\omega^\ast\cdot Q])=-\frac{1}{2}\alpha\lambda\chi^\gamma \notag
\end{align}
using the Girsanov transformation:
\begin{equation}
Z\to Z-\adja{\sigma}^\dagger\phi
\end{equation}
for the normal factors. Convexity between the final derivative and underlying prices cannot be eliminated by linear replication, which leaves residual risks in the portfolio when $\adja{\gamma}\ne0$. This penalty is not impacted by the hedge portfolio, as no hedge from the underlying market can replicate the quadratic factor dependence of the derivative.
\item[Dimensionality]Define the penalty $\chi^\sigma$ and attribution $X^\sigma$ for dimensionality:
\begin{align}
\chi^\sigma&:=\|(1-1_{\adjb{\sigma}})\adjb{\epsilon}\|^2 \\
X^\sigma&:=(\adja{\epsilon}-\adja{\gamma}\adja{\sigma}^\dagger\phi)\cdot(Z+\adja{\sigma}^\dagger\phi)-\delta^\sigma\cdot(Q-(1+r)q) \notag
\end{align}
The chain rule in the quadratic Gaussian theorem generates the dimensionality contribution to the entropy-adjusted mean:
\begin{align}
\frac{1}{\lambda}(&\mathbb{E}[\alpha][\omega^\ast\cdot Q+\lambda(X^\gamma+X^\sigma)] \\
&\qquad-\mathbb{E}[\alpha][\omega^\ast\cdot Q+\lambda X^\gamma]) \notag \\
&=\frac{1}{\lambda}\mathbb{E}[\alpha][\lambda(\adjb{\epsilon}-\adjb{\sigma}^\dagger\delta^\sigma)\cdot Z] \notag \\
&=-\frac{1}{2}\alpha\lambda\chi^\sigma \notag
\end{align}
using the Girsanov transformation:
\begin{equation}
Z\to(1+\alpha\lambda\adja{\gamma})^{-1/2}Z-\adja{\sigma}^\dagger\phi
\end{equation}
for the normal factors. The projection $1_{\adjb{\sigma}}$ projects the factors onto the hedgeable subspace. When $\adjb{\epsilon}$ is not in this subspace, so that it satisfies $(1-1_{\adjb{\sigma}})\adjb{\epsilon}\ne0$, there are residual unhedged factors that cannot be replicated.
\item[Funding]Define the penalty $\chi^s$ and attribution $X^s$ for funding:
\begin{align}
\chi^s&:=s^2\delta^s\cdot q \\
X^s&:=-s\delta^s\cdot(Q-(1+r)q) \notag
\end{align}
The chain rule in the quadratic Gaussian theorem generates the funding contribution to the entropy-adjusted mean:
\begin{align}
\frac{1}{\lambda}(&\mathbb{E}[\alpha][\omega^\ast\cdot Q+\lambda(X^\gamma+X^\sigma+X^s)] \\
&\qquad-\mathbb{E}[\alpha][\omega^\ast\cdot Q+\lambda(X^\gamma+X^\sigma)]) \notag \\
&=\frac{1}{\lambda}\mathbb{E}[\alpha][-\lambda s\adjb{\sigma}^\dagger\delta^s\cdot Z] \notag \\
&=-\frac{1}{2}\alpha\lambda\chi^s \notag
\end{align}
using the Girsanov transformation:
\begin{equation}
Z\to(1+\alpha\lambda\adja{\gamma})^{-1/2}(Z-\alpha\lambda(\adjb{\epsilon}-\adjb{\sigma}^\dagger\delta^\sigma))-\adja{\sigma}^\dagger\phi
\end{equation}
for the normal factors. The hedging leg offsets the derivative risks that are correlated with underlying risks. If the initial price of the hedging leg does not match the initial derivative price, so that $s\ne0$, the funding leg creates unhedged risks. This is the only penalty that depends on the initial derivative price.
\item[Strategy]Define the penalty $\chi^\delta$ and attribution $X^\delta$ for strategy:
\begin{align}
\chi^\delta&:=\|\adjb{\sigma}^\dagger(\delta-\delta^\ast)\|^2 \\
X^\delta&:=-(\delta-\delta^\ast)\cdot Q \notag
\end{align}
The chain rule in the quadratic Gaussian theorem generates the strategy contribution to the entropy-adjusted mean:
\begin{align}
\frac{1}{\lambda}(&\mathbb{E}[\alpha][\omega^\ast\cdot Q+\lambda(X^\gamma+X^\sigma+X^s+X^\delta)] \\
&\qquad-\mathbb{E}[\alpha][\omega^\ast\cdot Q+\lambda(X^\gamma+X^\sigma+X^s)]) \notag \\
&=\frac{1}{\lambda}\mathbb{E}[\alpha][-\lambda\adjb{\sigma}^\dagger(\delta-\delta^\ast)\cdot Z] \notag \\
&=-\frac{1}{2}\alpha\lambda\chi^\delta \notag
\end{align}
using the Girsanov transformation:
\begin{equation}
Z\to(1+\alpha\lambda\adja{\gamma})^{-1/2}(Z-\alpha\lambda(\adjb{\epsilon}-\adjb{\sigma}^\dagger\delta^\ast))-\adja{\sigma}^\dagger\phi
\end{equation}
for the normal factors. If the investor funds the derivative suboptimally, so that $\delta\ne\delta^\ast$, the entropy-adjusted mean incurs an additional penalty. This is the only penalty that can be controlled by the investor, and is eliminated when the funding strategy matches the optimal hedge portfolio.
\end{description}

Entropic risk in the quadratic Gaussian model is thus decomposed via a chain of four consecutive measure changes, each of which introduces a source of unhedged risk that is attributed in the realised P\&L. The expressions in this subsection are valid for the convergence range of the quadratic Gaussian theorem, with $\alpha\lambda$ in an open interval containing zero. The mid solution and bid-offer corrections for small notional are obtained from the first-order expansion around $\alpha\lambda=0$.

\subsection*{Pricing}

Using the expressions above, the maximum entropy-adjusted mean for the funded derivative is obtained by eliminating the strategy penalty with the choice $\delta^\ast$ for the hedge portfolio. The incremental entropy-adjusted mean generated from this choice is:
\begin{align}
&\frac{1}{\lambda}(\mathbb{E}[\alpha][\omega^\ast\cdot Q+\lambda(P-\delta^\ast\cdot Q)]-\mathbb{E}[\alpha][\omega^\ast\cdot Q])= \\
&\qquad\zeta-\frac{(1+r+\alpha\lambda s)^2-(1+r)^2}{2\alpha\lambda}\delta^s\cdot q \notag
\end{align}
where:
\begin{equation}
\zeta:=(\adjb{\eta}-(1+r)\delta^\sigma\cdot q)-\frac{1}{2}\alpha\lambda(\chi^\gamma+\chi^\sigma)
\end{equation}
is the contribution that does not depend on $s$. The indifference price is derived by setting the increment to zero, meaning that the derivative neither increases nor decreases the optimal performance that can be achieved from the strategy. This generates the solution:
\begin{align}
s^\ast&:=\frac{1+r}{\alpha\lambda}\left(\sqrt{1+\frac{2\alpha\lambda\zeta}{(1+r)^2\delta^s\cdot q}}-1\right) \\
p^\ast&:=\frac{1}{1+r}\left(\adjb{\eta}-\frac{1}{2}\alpha\lambda(\chi^\gamma+\chi^\sigma+{s^\ast}^2\delta^s\cdot q)\right) \notag
\end{align}
for the indifference funding spread $s^\ast$ and initial price $p^\ast$ of the derivative, choosing the quadratic branch that converges as $\alpha\lambda\to0$. The mid indifference price is the mean of the final price adjusted for measure and convexity, and the penalties generate the bid-offer range for long and short positions, all discounted at the funding rate. It is valid for $\alpha\lambda$ satisfying:
\begin{equation}
\alpha\lambda\zeta\ge-\frac{1}{2}(1+r)^2\delta^s\cdot q
\end{equation}
intersected with the domain of validity of the quadratic Gaussian theorem, which includes $\alpha\lambda=0$. Outside this range, there is no indifference price for which the addition of the hedged derivative matches the optimal performance in the underlying market.

\begin{figure*}[!tp]
\centering\normalsize\setlength{\tabcolsep}{0.5\columnsep}
\begin{tabular}{@{}C{0.39\textwidth-\tabcolsep}C{0.39\textwidth-\tabcolsep}@{}}
\includegraphics[width=0.39\textwidth-\tabcolsep]{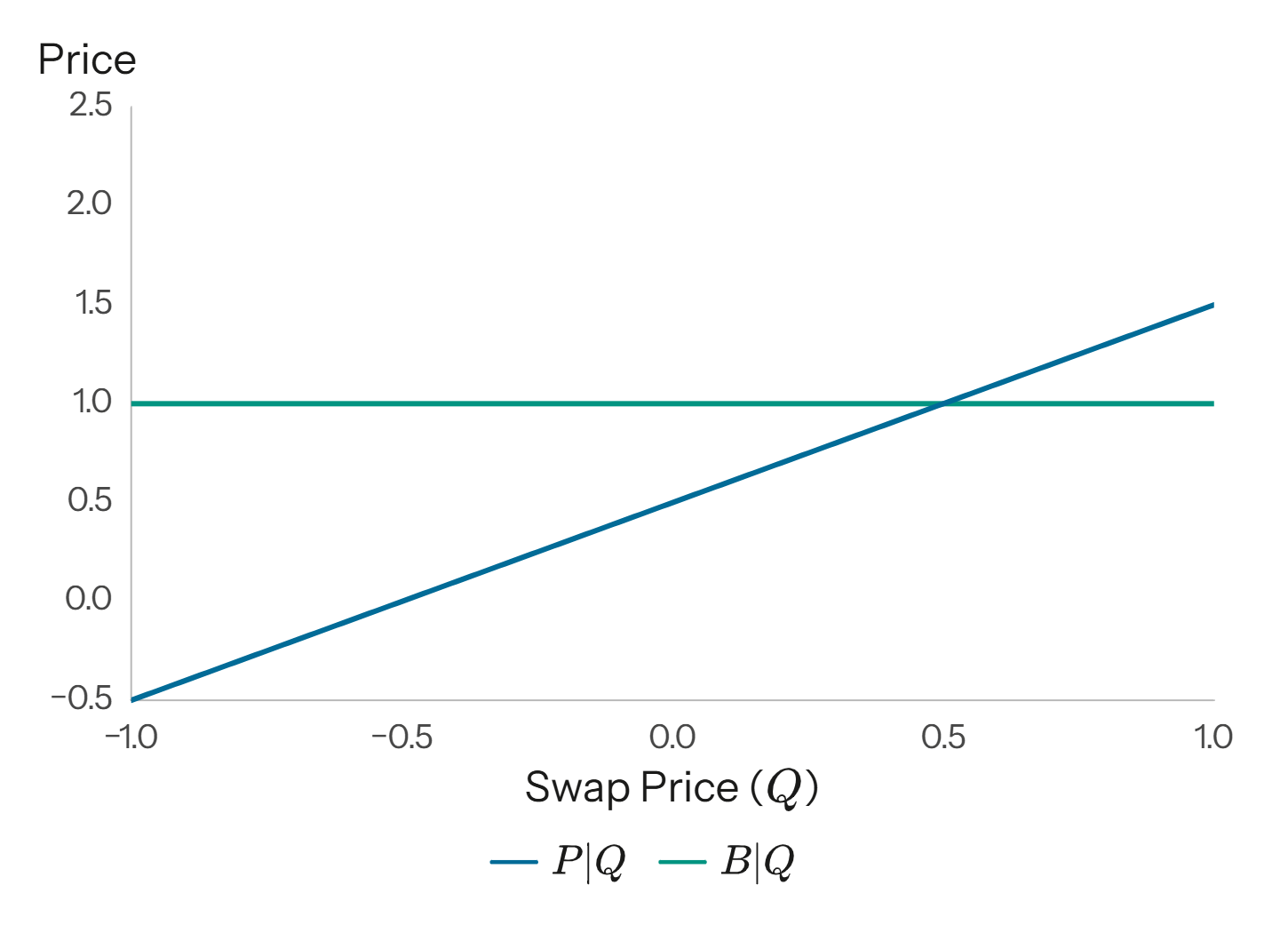}&
\includegraphics[width=0.39\textwidth-\tabcolsep]{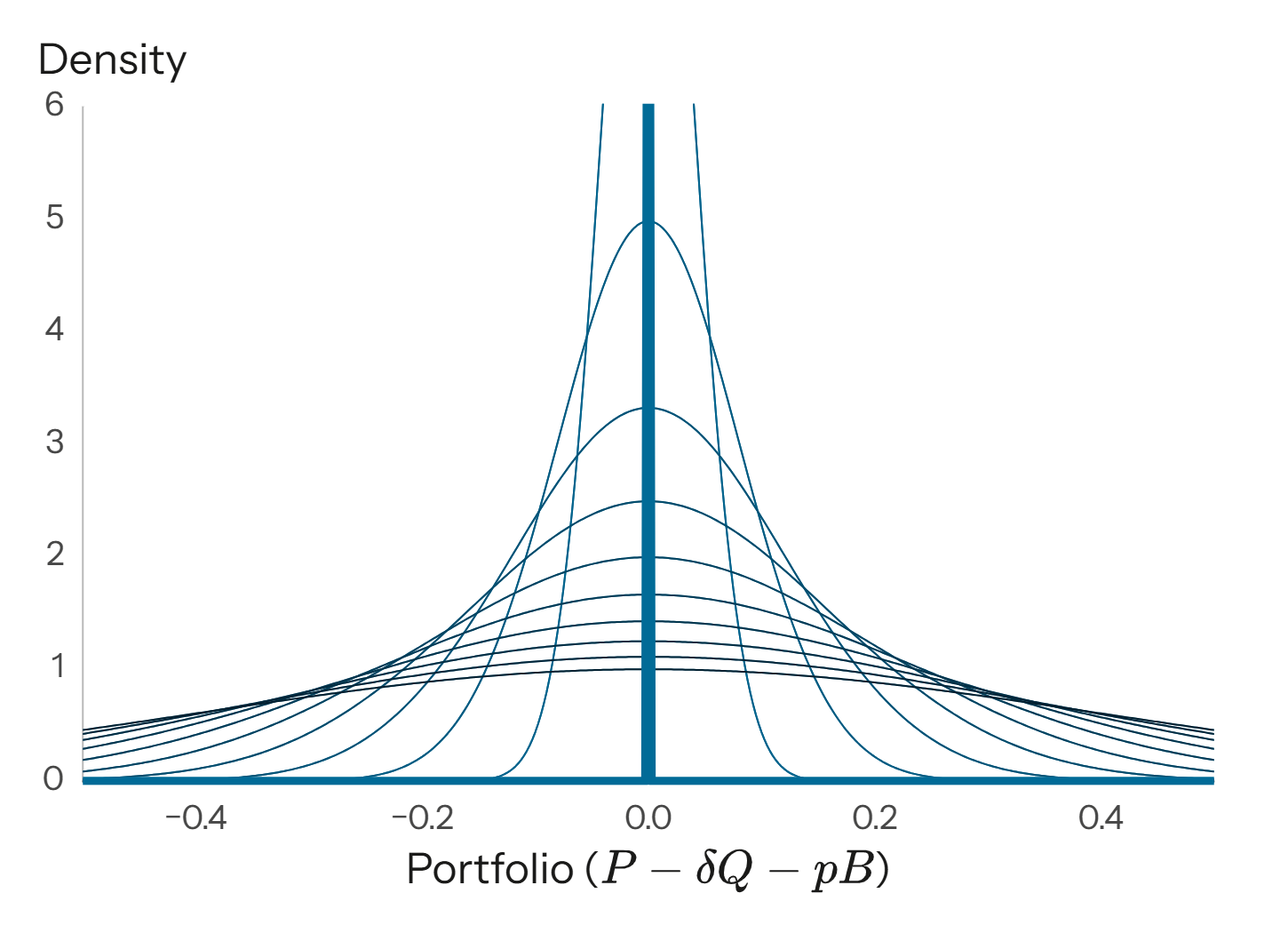}\\
\multicolumn{2}{c}{Case $\gamma=0$, $\rho=0$, $\sigma_b=0$}\\
\includegraphics[width=0.39\textwidth-\tabcolsep]{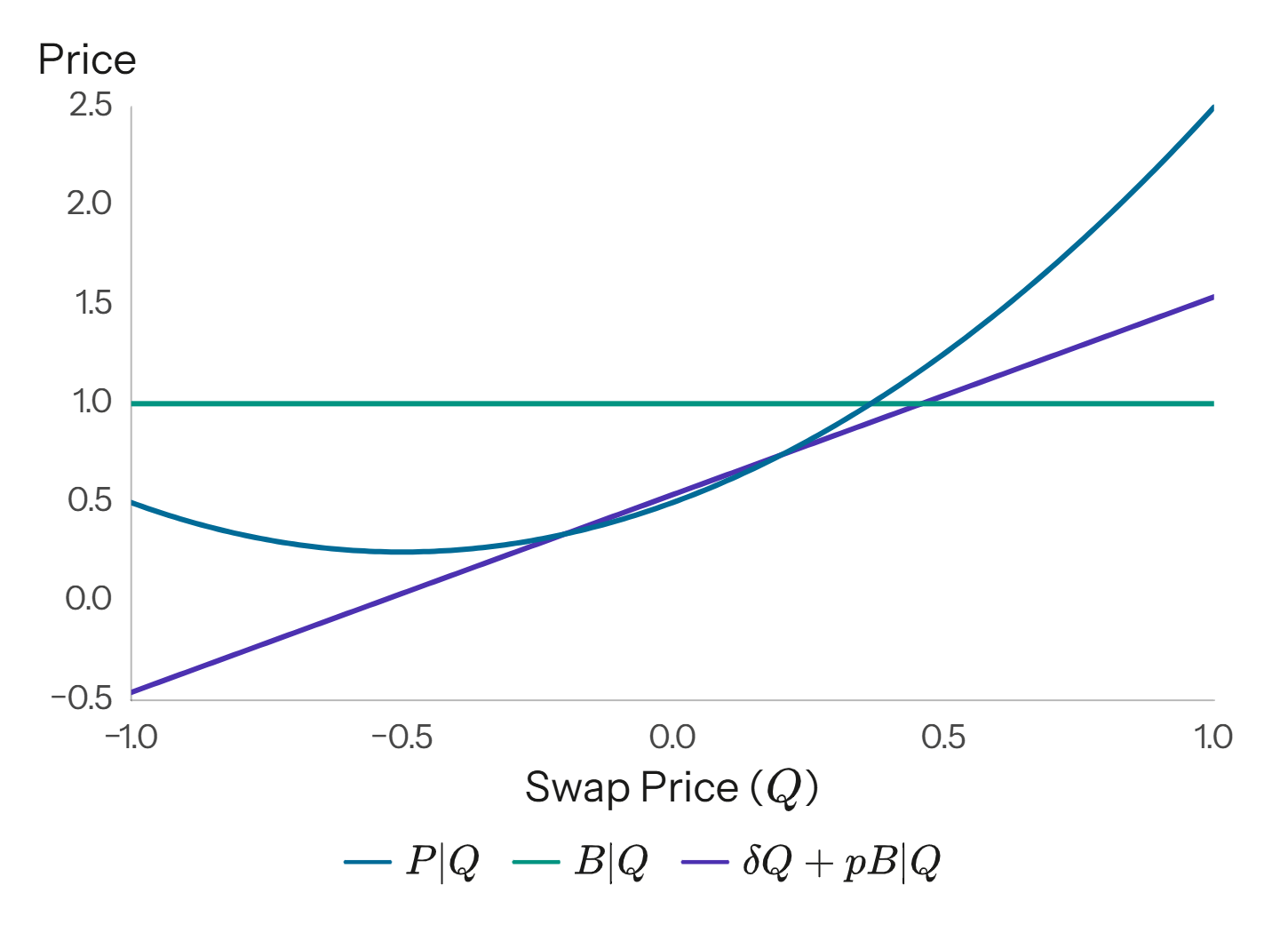}&
\includegraphics[width=0.39\textwidth-\tabcolsep]{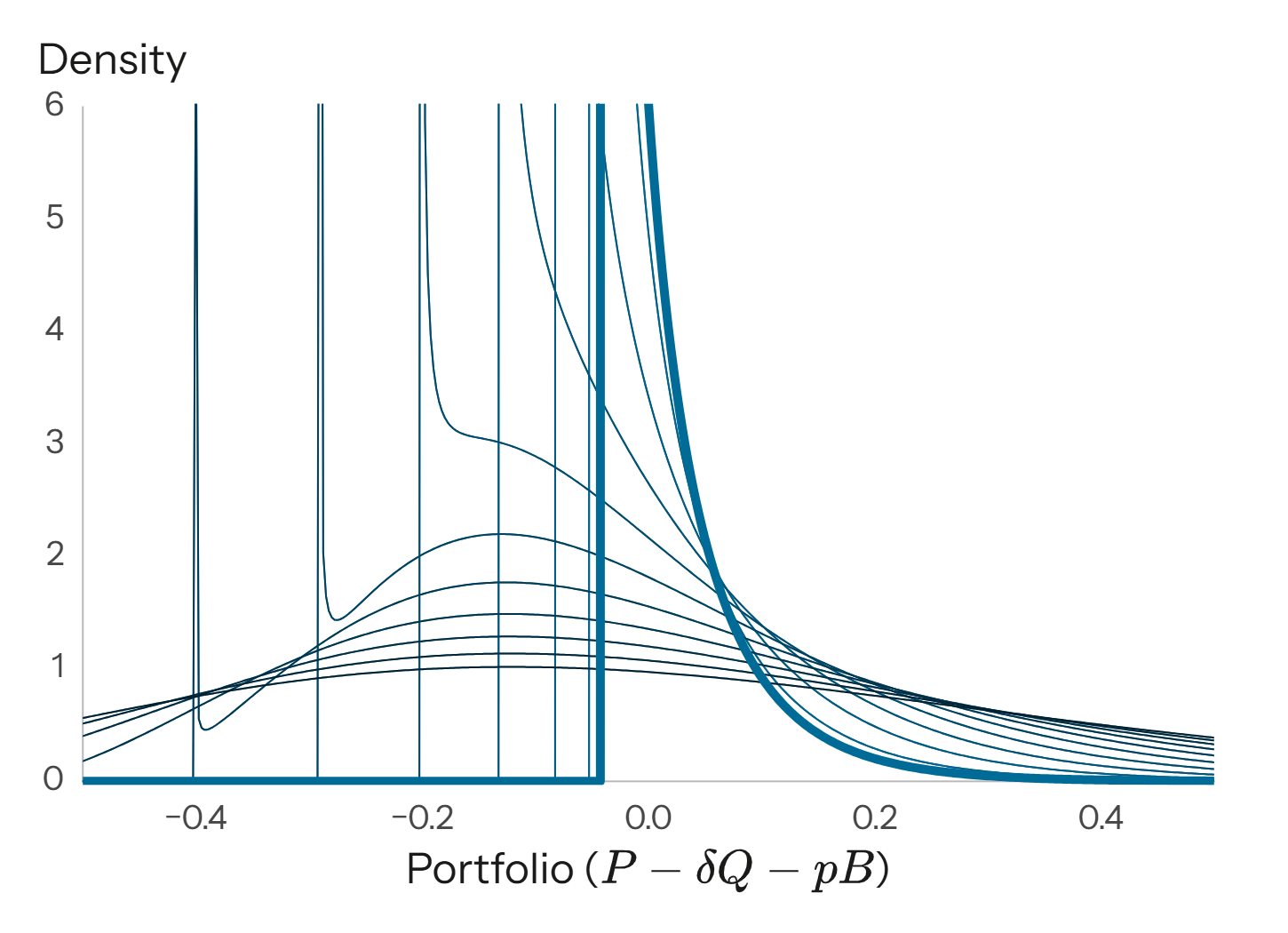}\\
\multicolumn{2}{c}{Case $\gamma=2$, $\rho=0$, $\sigma_b=0$}\\
\includegraphics[width=0.39\textwidth-\tabcolsep]{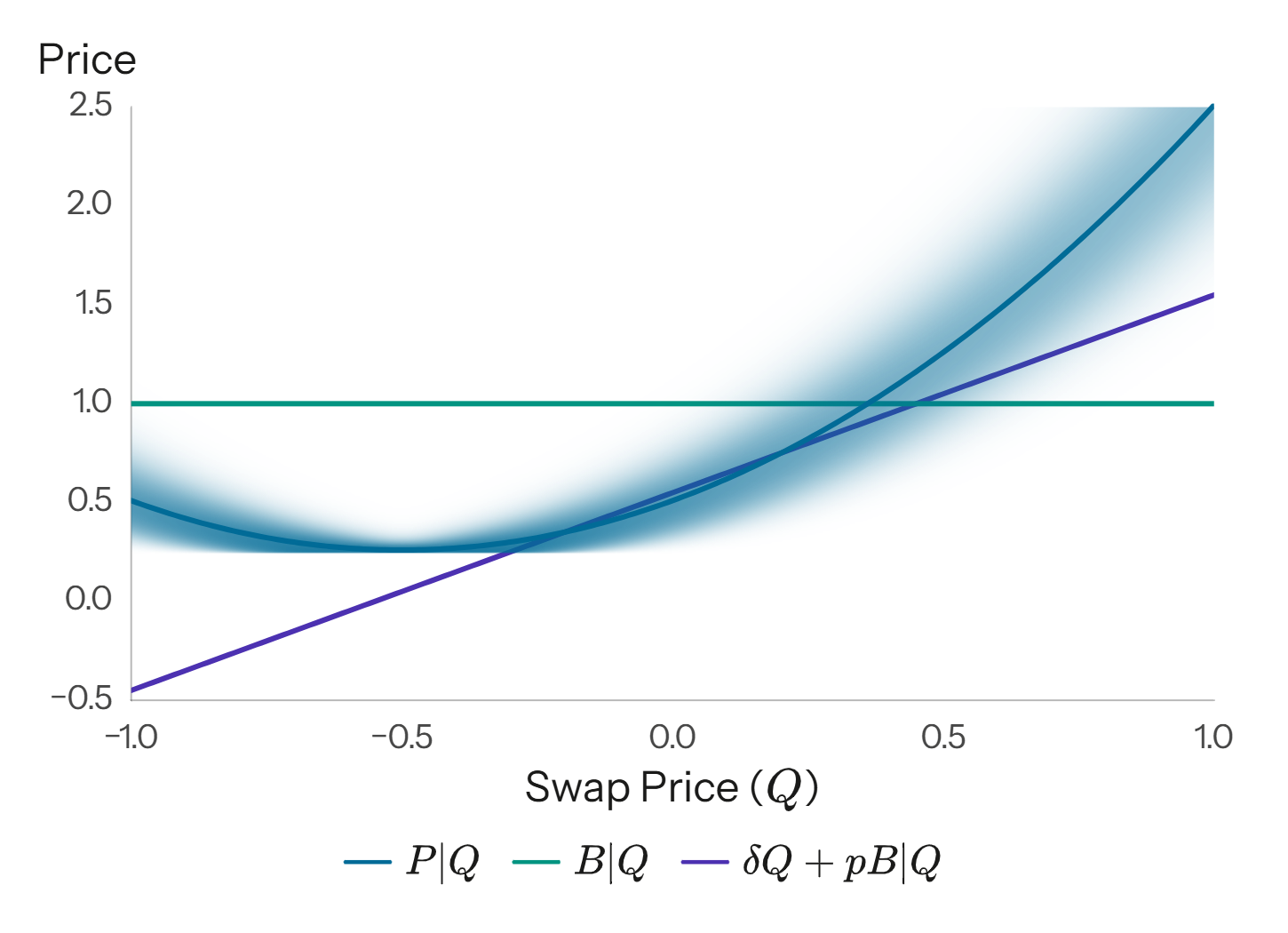}&
\includegraphics[width=0.39\textwidth-\tabcolsep]{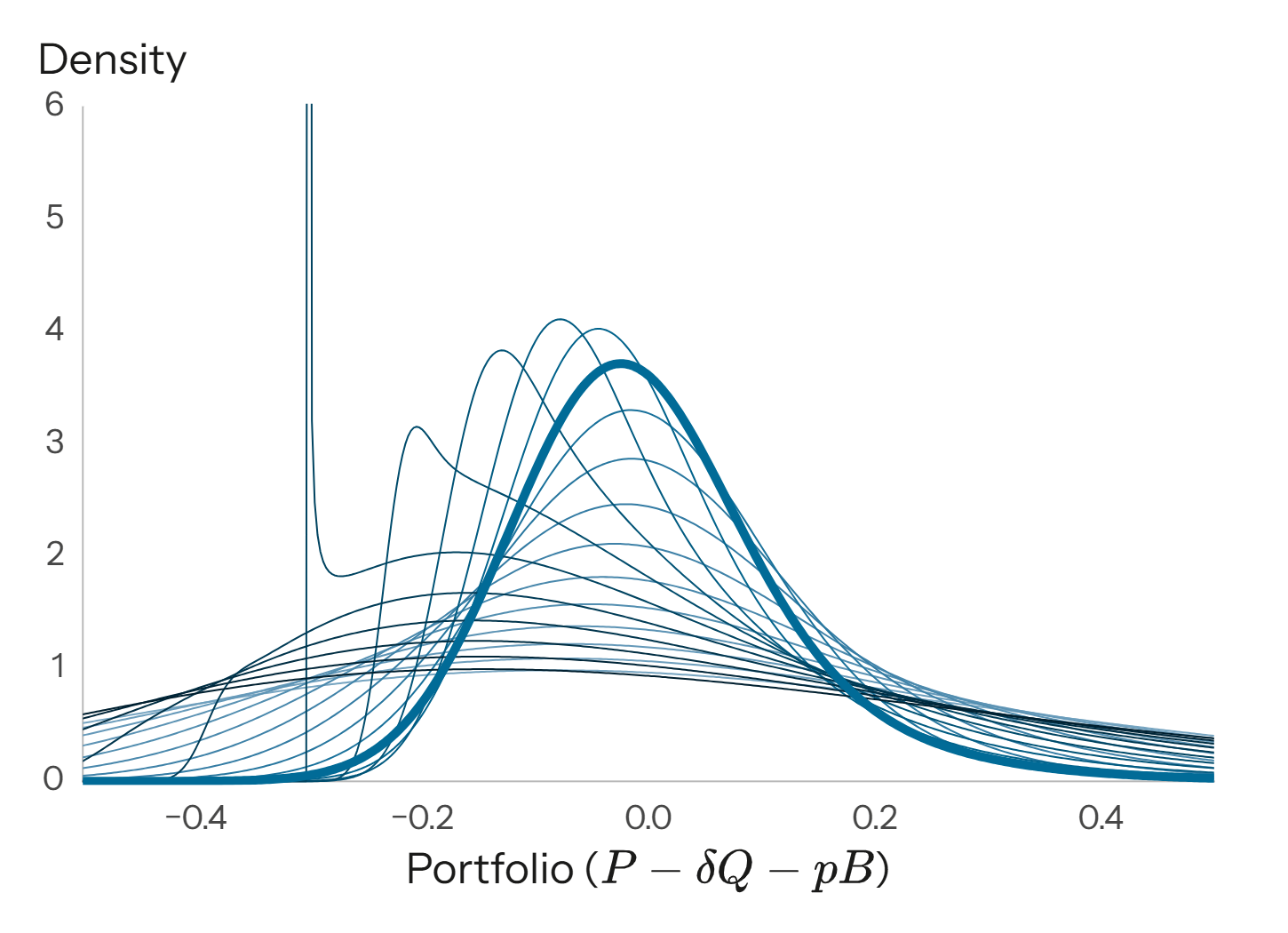}\\
\multicolumn{2}{c}{Case $\gamma=2$, $\rho=0.5$, $\sigma_b=0$}\\
\includegraphics[width=0.39\textwidth-\tabcolsep]{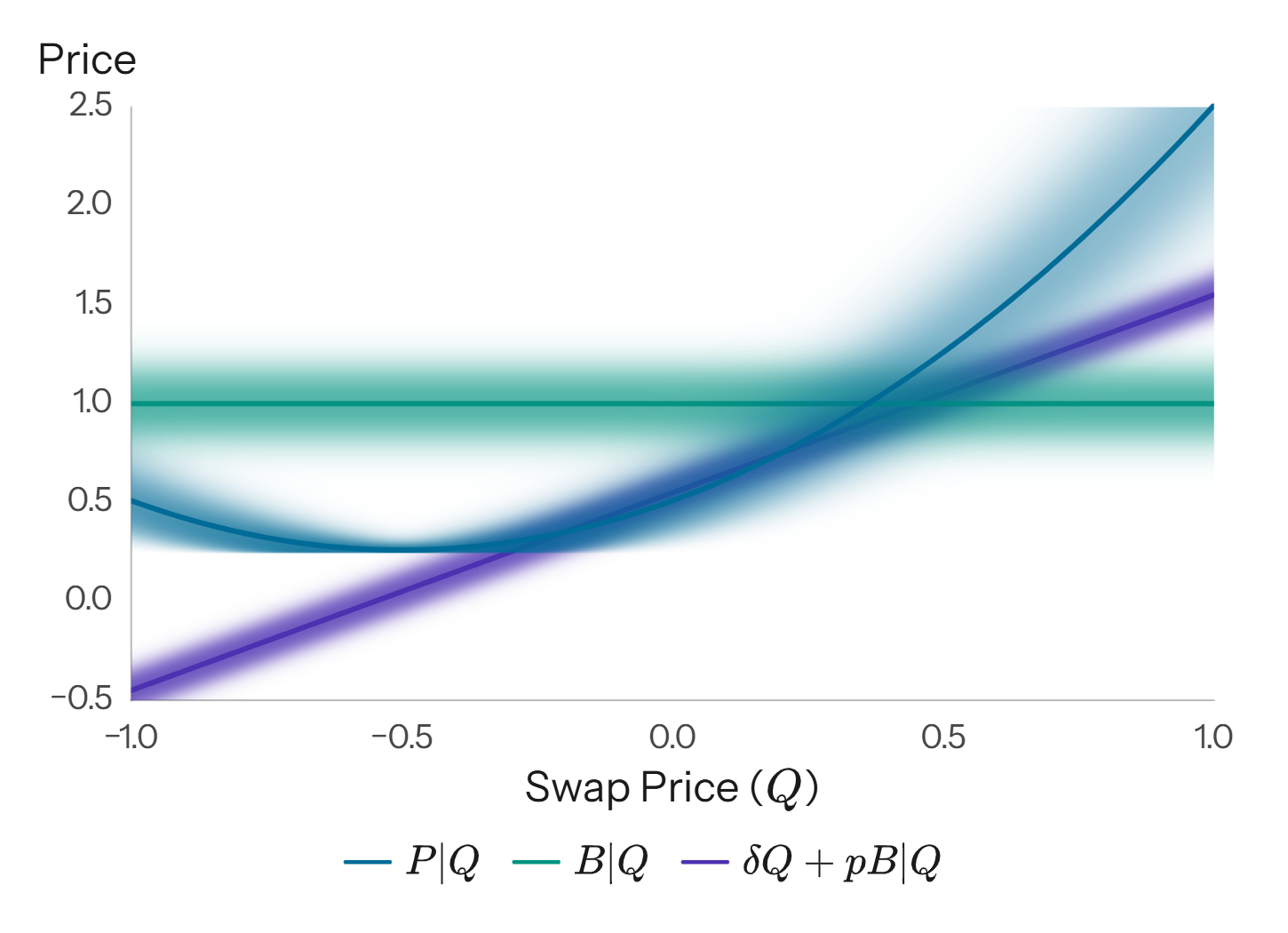}&
\includegraphics[width=0.39\textwidth-\tabcolsep]{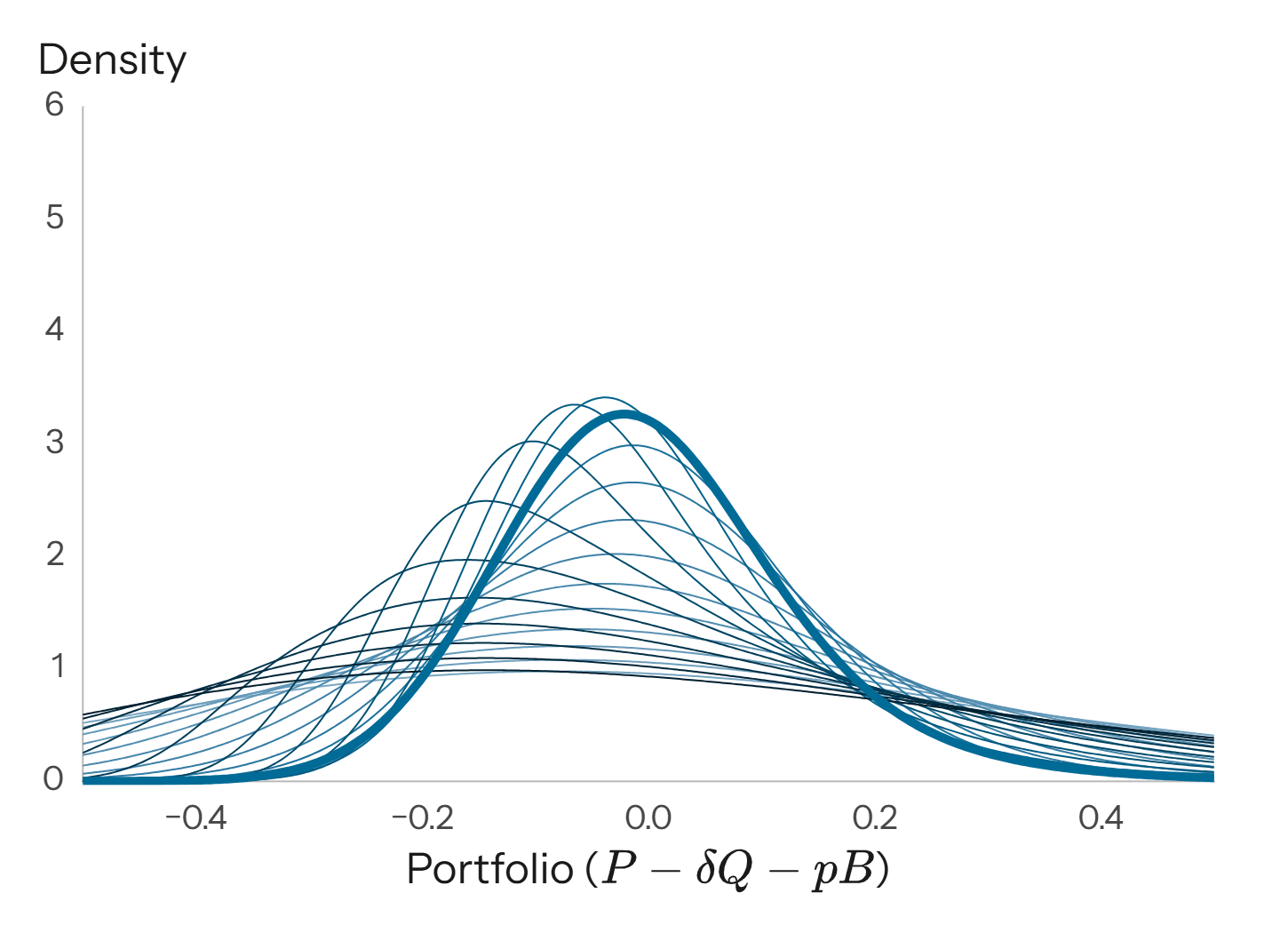}\\
\multicolumn{2}{c}{Case $\gamma=2$, $\rho=0.5$, $\sigma_b=0.1$}
\end{tabular}
\caption{Quadratic Gaussian P\&L is a combination of normal and noncentral chi-squared variables. In this case study, optimal hedging in the mid model ($\alpha\lambda=0$) is identified by varying the hedge ratio across the range $-1\le\delta\le3$, with the optimal choice at $\delta^\ast=1$ highlighted. As represented in the graphs of the left column for the optimal hedge $\delta=\delta^\ast$, introducing market incompleteness via convexity, dimensionality and funding contributions reduces the effectiveness of replication.}
\end{figure*}

\section*{Case study}

As a case study, consider an underlying market comprising a bond and a swap. The bond is used for funding, and has initial price $b=1$ and lognormal final price approximated as $B=1+\sigma_bZ_b$ where $\sigma_b$ is the lognormal volatility. The swap is used for hedging, and has initial price $q=0$ and normal final price $Q=\sigma_qZ_q$ where $\sigma_q$ is the normal volatility. The final price of the derivative is assumed to be quadratic in the final price of the swap with additional noise:
\begin{equation}
P=\eta+\epsilon(Q+\rho\sigma_qZ)+\frac{1}{2}\gamma(Q+\rho\sigma_qZ)^2
\end{equation}
parametrised by the scalars $\eta$, $\epsilon$, $\gamma$ and $\rho$. The variables $Z_b$, $Z_q$ and $Z$ are independent standard normal variables. This simple setup demonstrates the key features of entropic pricing in the quadratic Gaussian model. The scenario where the derivative can be exactly replicated corresponds to $\gamma=0$ (no convexity penalty), $\rho=0$ (no dimensionality penalty) and $\sigma_b=0$ (no funding penalty), and dialling parameters away from this scenario generates the three incompleteness penalties and their contributions to pricing and hedging.

\begin{figure*}[!tp]
\centering\normalsize\setlength{\tabcolsep}{0.5\columnsep}
\begin{tabular}{@{}C{0.32\textwidth-\tabcolsep}C{0.32\textwidth-\tabcolsep}C{0.32\textwidth-\tabcolsep}@{}}
\includegraphics[width=0.32\textwidth-\tabcolsep]{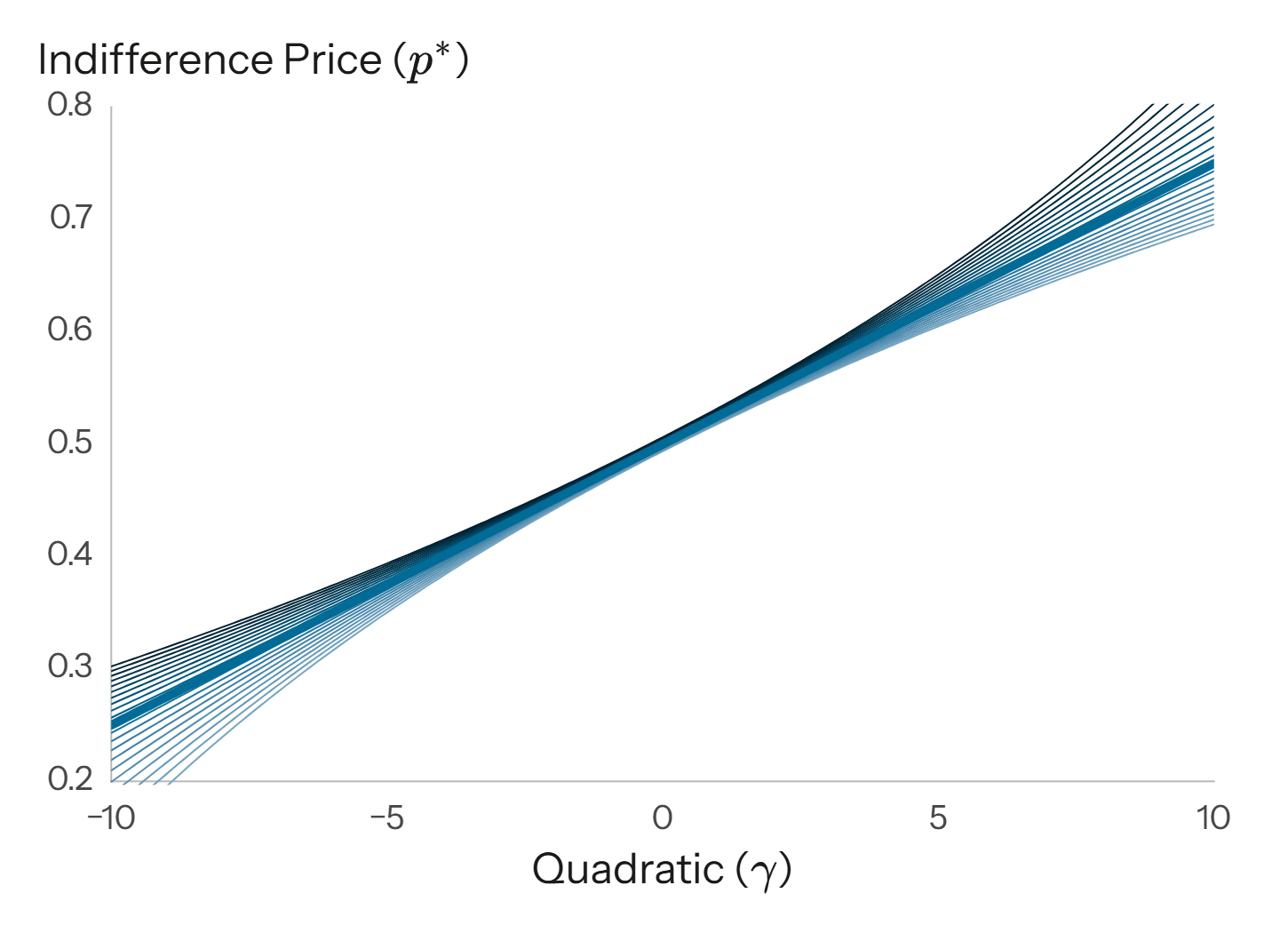}&
\includegraphics[width=0.32\textwidth-\tabcolsep]{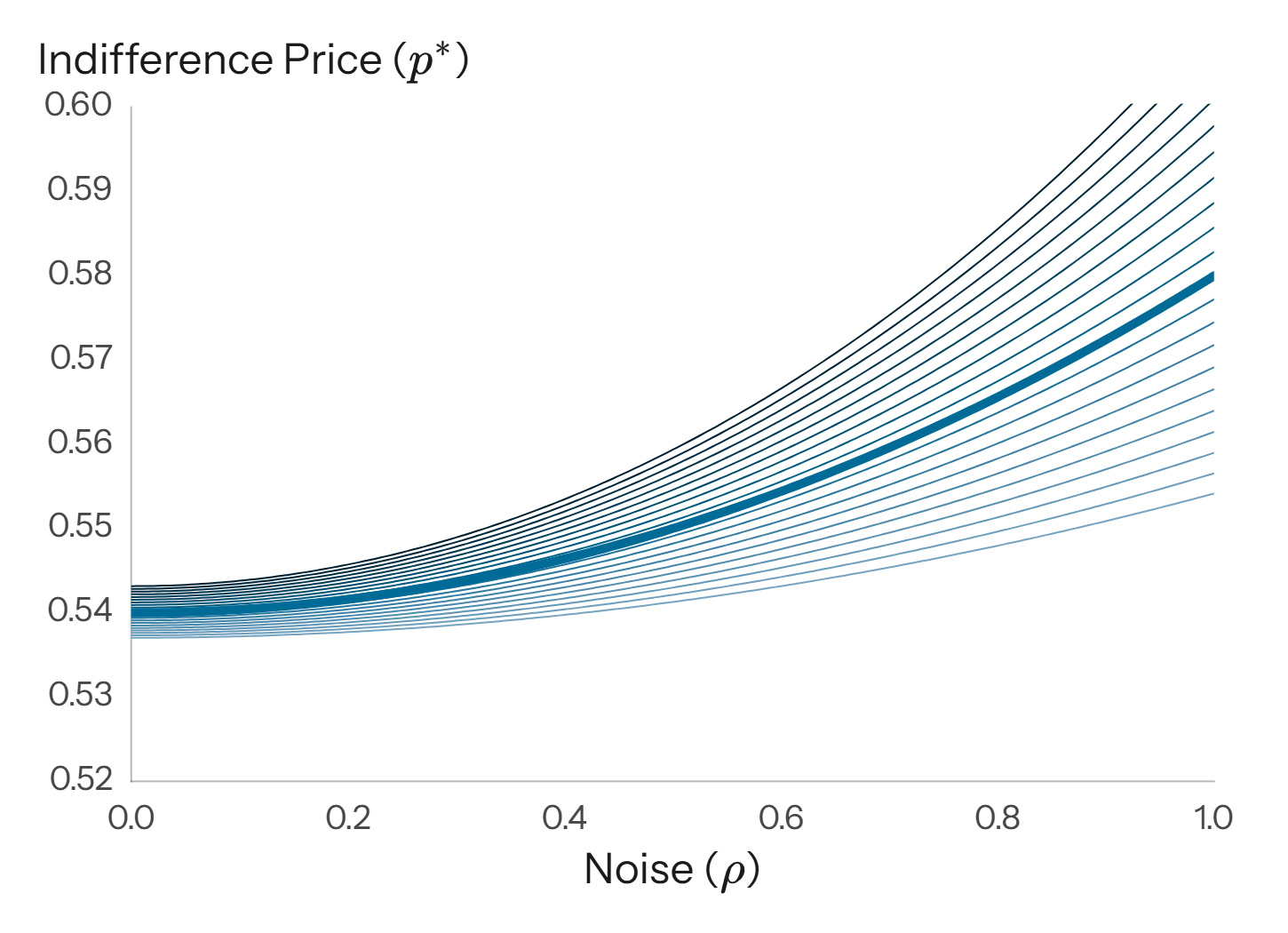}&
\includegraphics[width=0.32\textwidth-\tabcolsep]{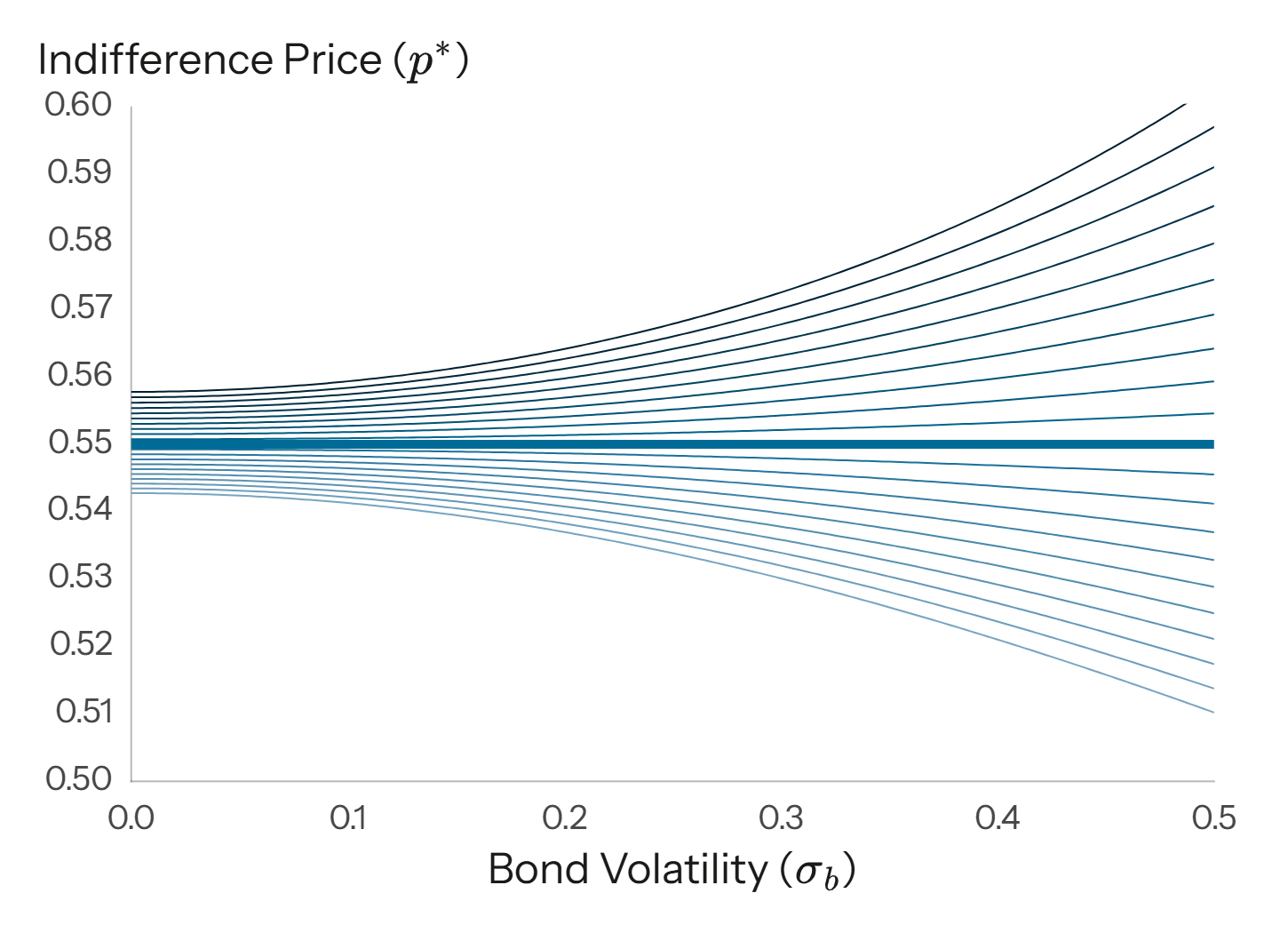}\\
\includegraphics[width=0.32\textwidth-\tabcolsep]{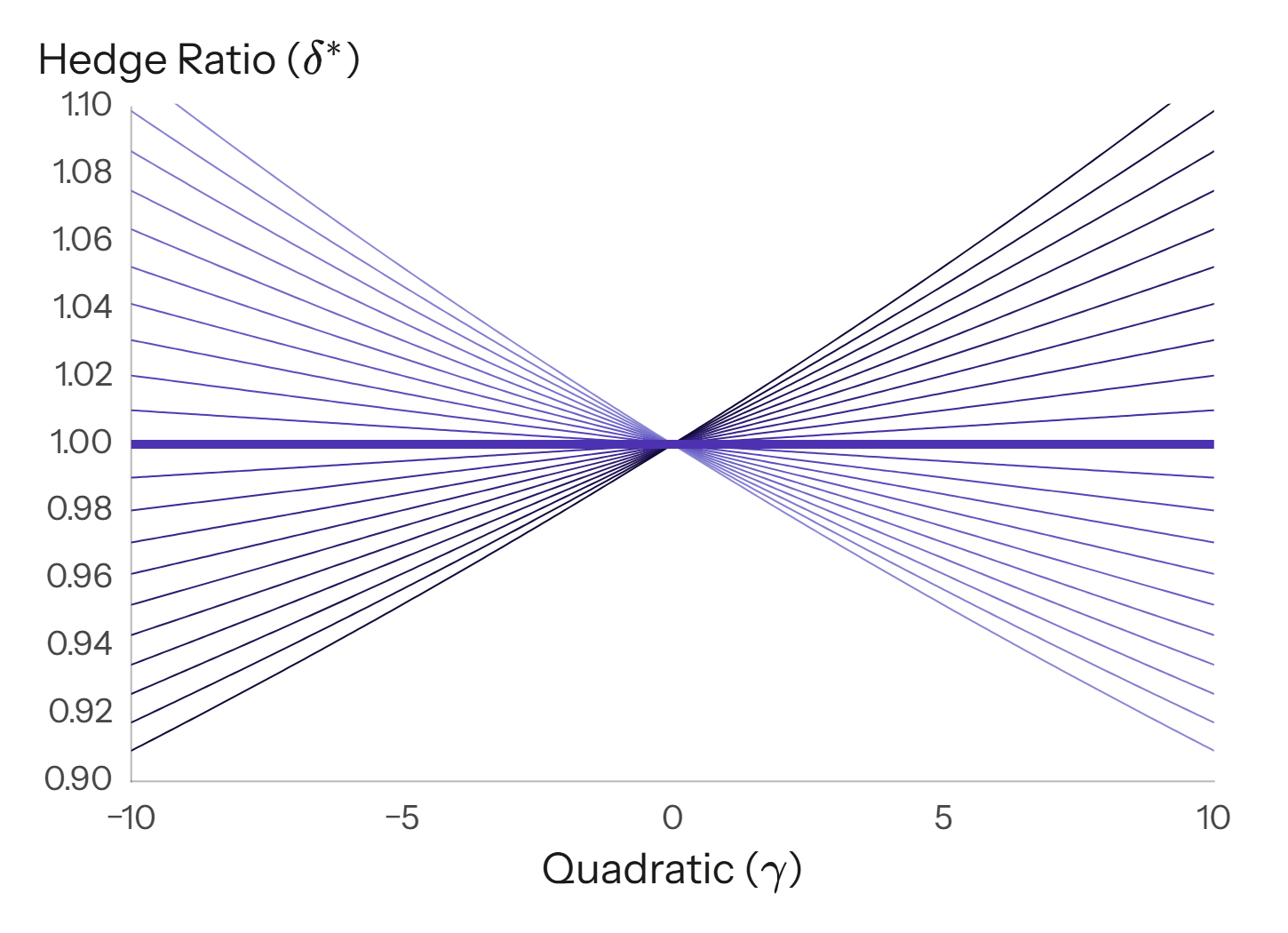}&
\includegraphics[width=0.32\textwidth-\tabcolsep]{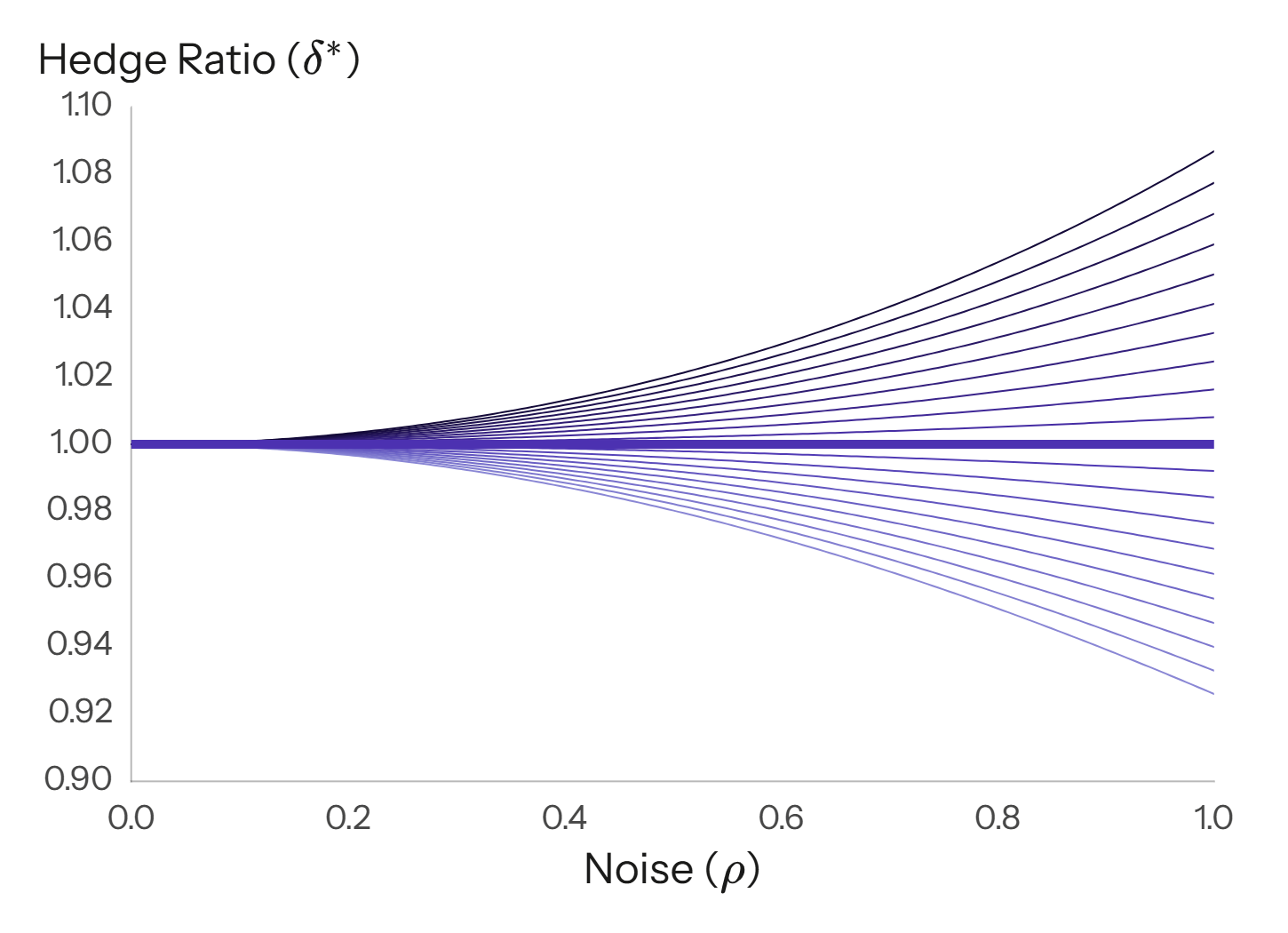}&
\includegraphics[width=0.32\textwidth-\tabcolsep]{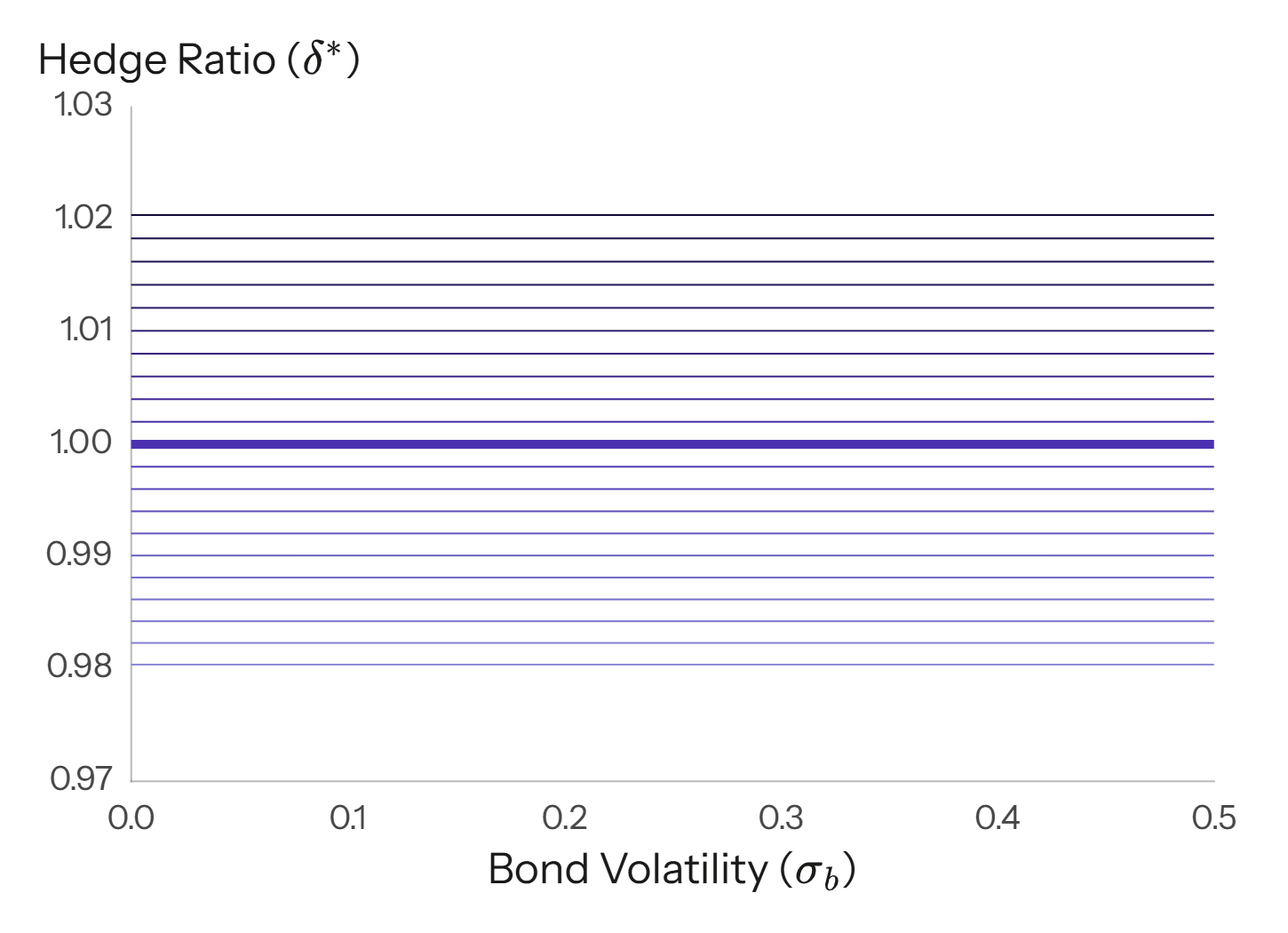}\\
Convexity&Dimensionality&Funding
\end{tabular}
\caption{Impact of the quadratic parameter $\gamma$, noise parameter $\rho$, and bond volatility $\sigma_b$ on the indifference price and hedge ratio for notionals in the range $-1\le\alpha\lambda\le1$, highlighting the mid case $\alpha\lambda=0$. For the convexity graphs, long/short positions generate the lower/upper curves in the indifference price plot and the decreasing/increasing curves in the hedge ratio plot. For the dimensionality and funding graphs, long/short positions generate the lower/upper curves in both plots. Residual risk from market incompleteness is mitigated with adjustments to the price as a model risk buffer and to the hedge ratio as protection against downside risk.}
\end{figure*}

The self-funding condition simplifies as only the bond is available for funding. The self-funded underlying-only portfolio comprises any notional $\omega$ of the swap and none of the bond, with entropy-adjusted mean:
\begin{equation}
\mathbb{E}[\alpha][\omega Q]=-\frac{1}{2}\alpha\sigma_q^2\omega^2
\end{equation}
maximised at the optimal portfolio $\omega^\ast=0$. The self-funded derivative and underlying portfolio comprises one of the derivative, any notional $-\delta$ of the swap, and notional $-p$ of the bond, per unit of derivative notional. When $\gamma\ne0$ the final price is:
\begin{align}
P-\delta Q&-p B=c-\sqrt{\frac{\rho^2\sigma_q^2\delta^2}{1+\rho^2}+\sigma_b^2p^2}\,U \\
&+\frac{1}{2}\gamma(1+\rho^2)\sigma_q^2\left(V+\frac{\epsilon(1+\rho^2)-\delta}{\gamma(1+\rho^2)^{3/2}\sigma_q}\right)^{\!2} \notag
\end{align}
where:
\begin{align}
U&:=\frac{\dfrac{\rho\sigma_q\delta}{1+\rho^2}(\rho Z_q-Z)+\sigma_bpZ_b}{\sqrt{\dfrac{\rho^2\sigma_q^2\delta^2}{1+\rho^2}+\sigma_b^2p^2}} \\
V&:=\frac{Z_q+\rho Z}{\sqrt{1+\rho^2}} \notag
\end{align}
are independent standard normal variables, and the constant is:
\begin{equation}
c:=\eta-p-\frac{1}{2}\frac{(\epsilon(1+\rho^2)-\delta)^2}{\gamma(1+\rho^2)^2}
\end{equation}
The final price is thus the linear combination of two independent variables, one standard normal and one noncentral chi-squared with one degree of freedom, reducing to a normal variable when $\gamma=0$.

\begin{figure*}[!tp]
\centering\normalsize\setlength{\tabcolsep}{0.5\columnsep}
\begin{tabular}{@{}C{0.32\textwidth-\tabcolsep}C{0.32\textwidth-\tabcolsep}C{0.32\textwidth-\tabcolsep}@{}}
\includegraphics[width=0.32\textwidth-\tabcolsep]{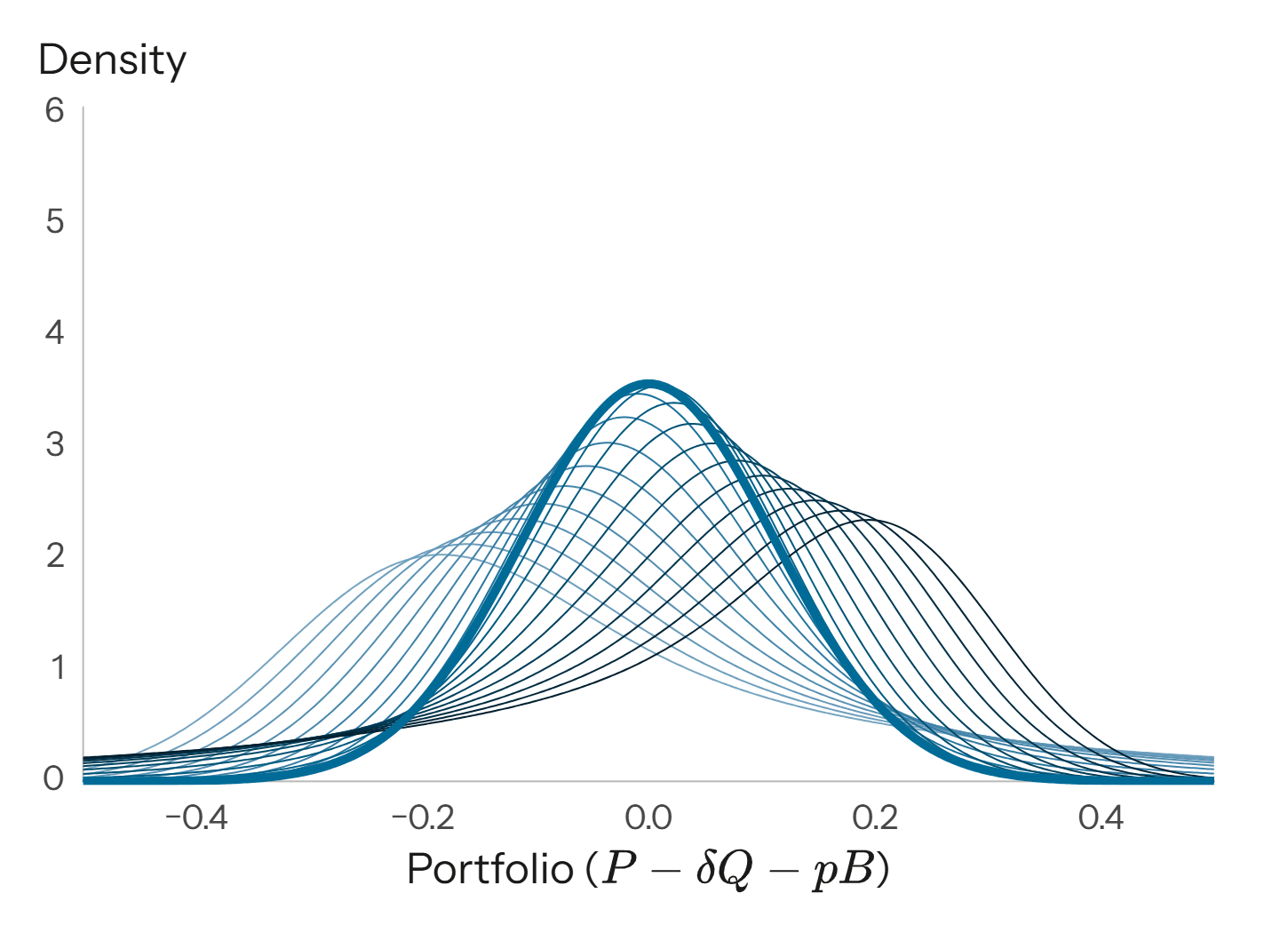}&
\includegraphics[width=0.32\textwidth-\tabcolsep]{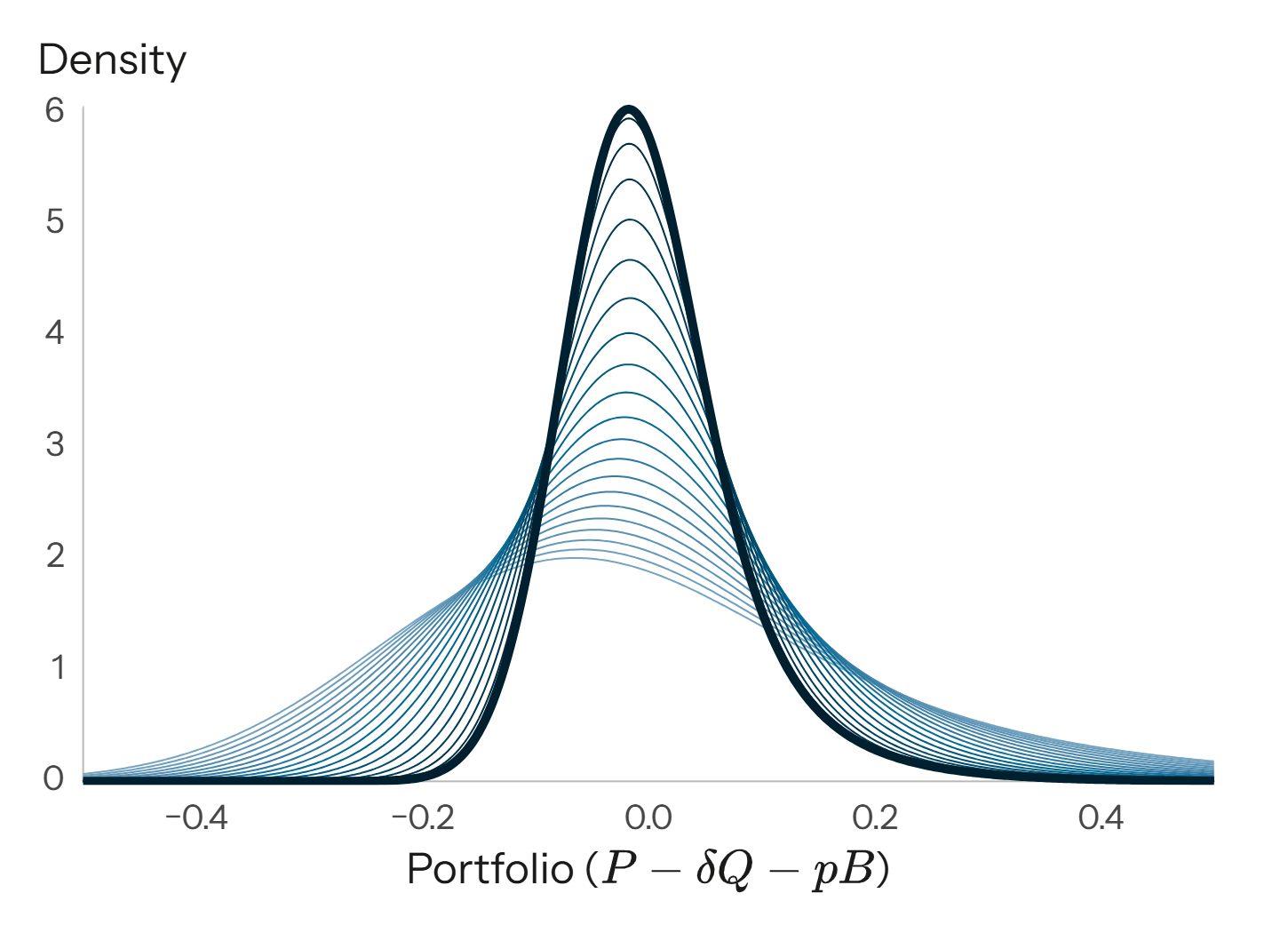}&
\includegraphics[width=0.32\textwidth-\tabcolsep]{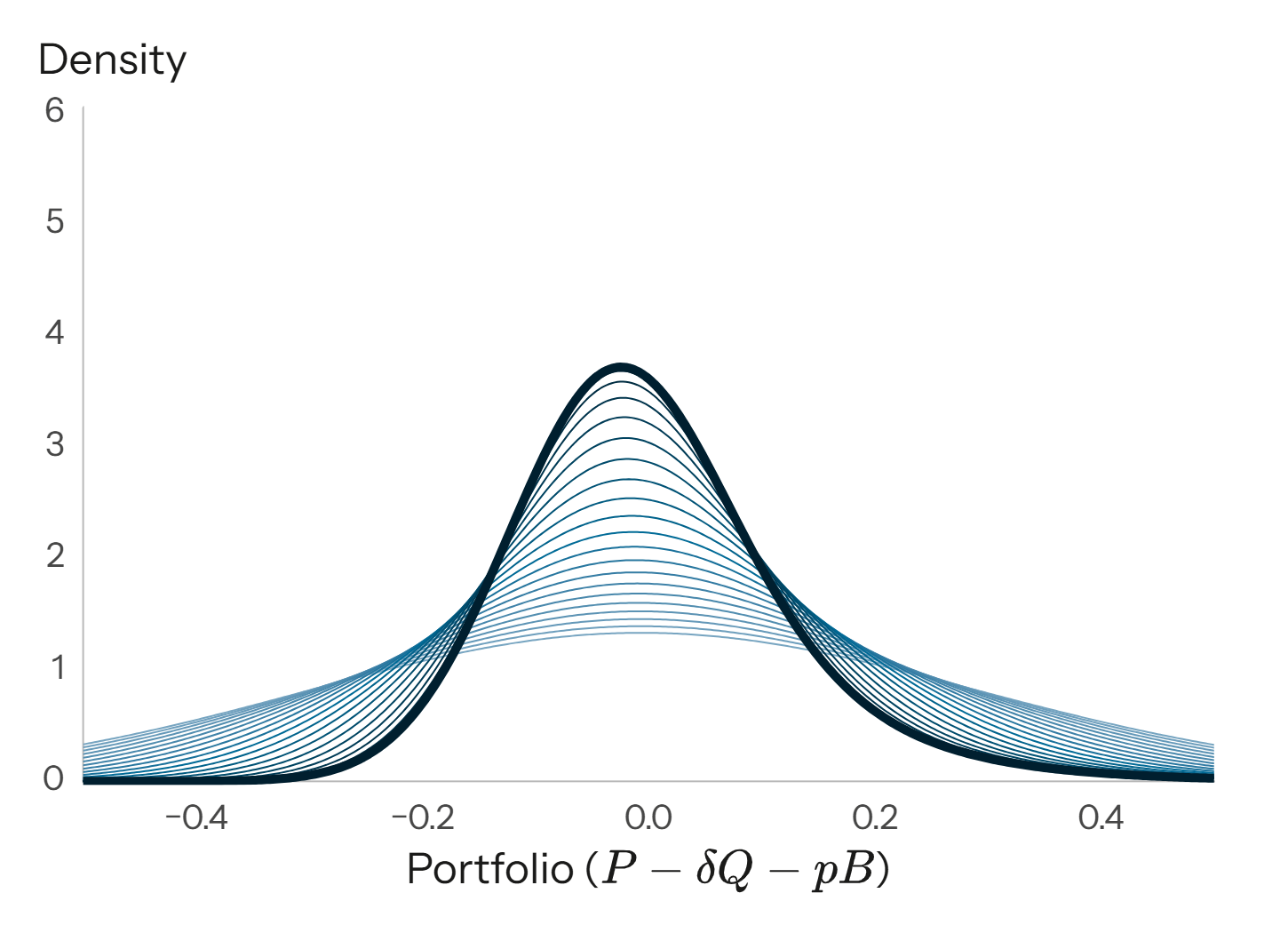}\\
\includegraphics[width=0.32\textwidth-\tabcolsep]{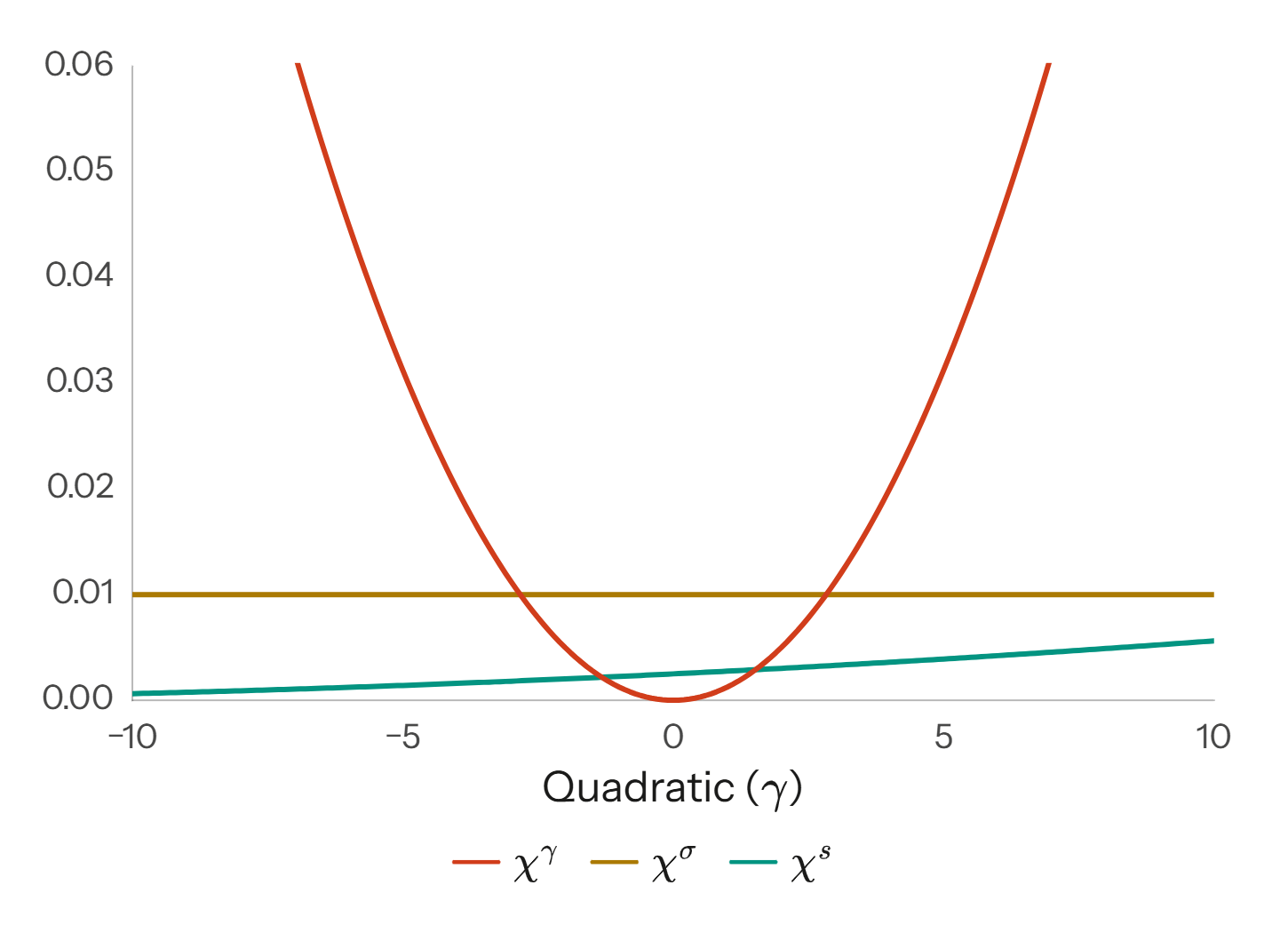}&
\includegraphics[width=0.32\textwidth-\tabcolsep]{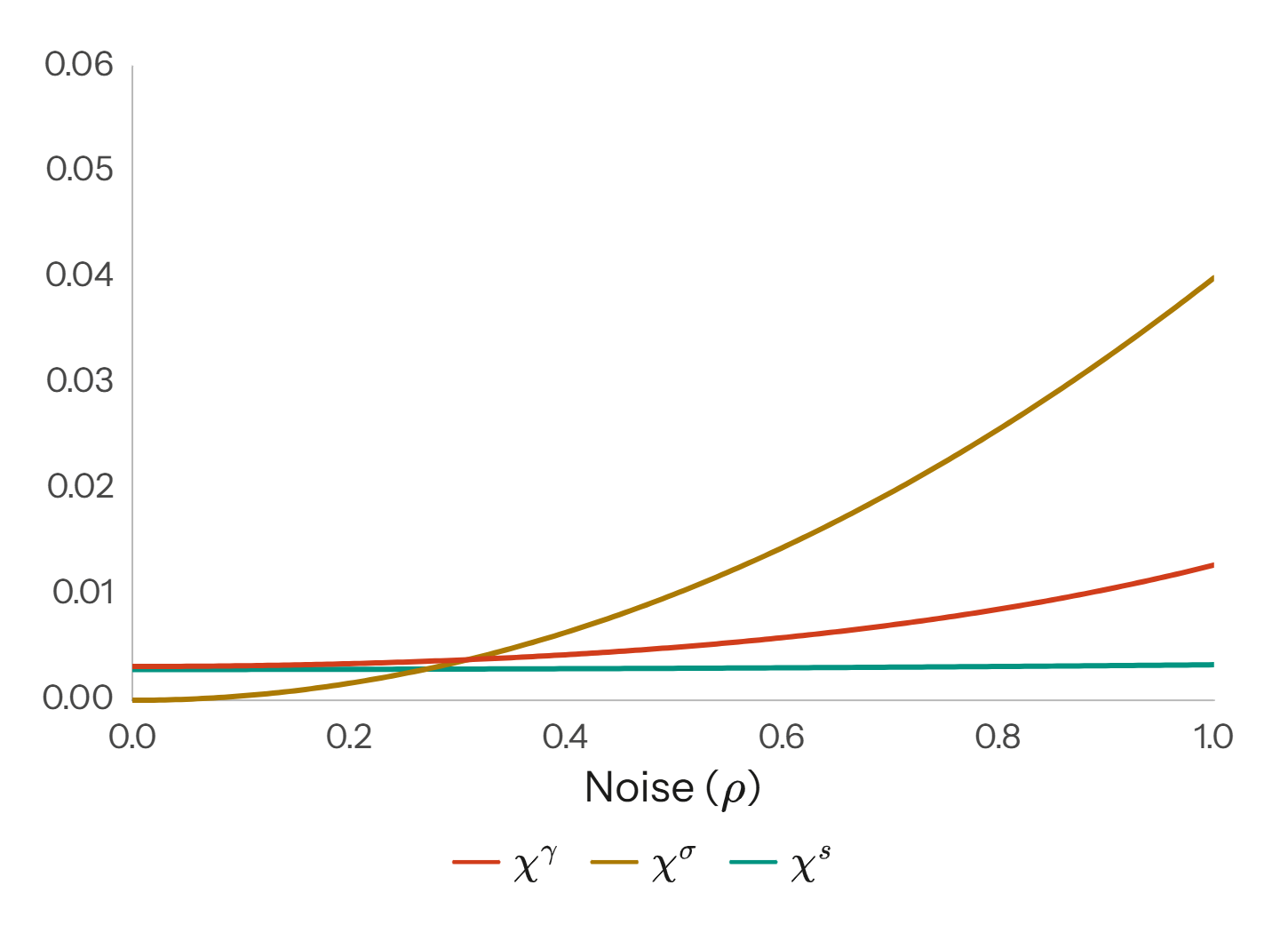}&
\includegraphics[width=0.32\textwidth-\tabcolsep]{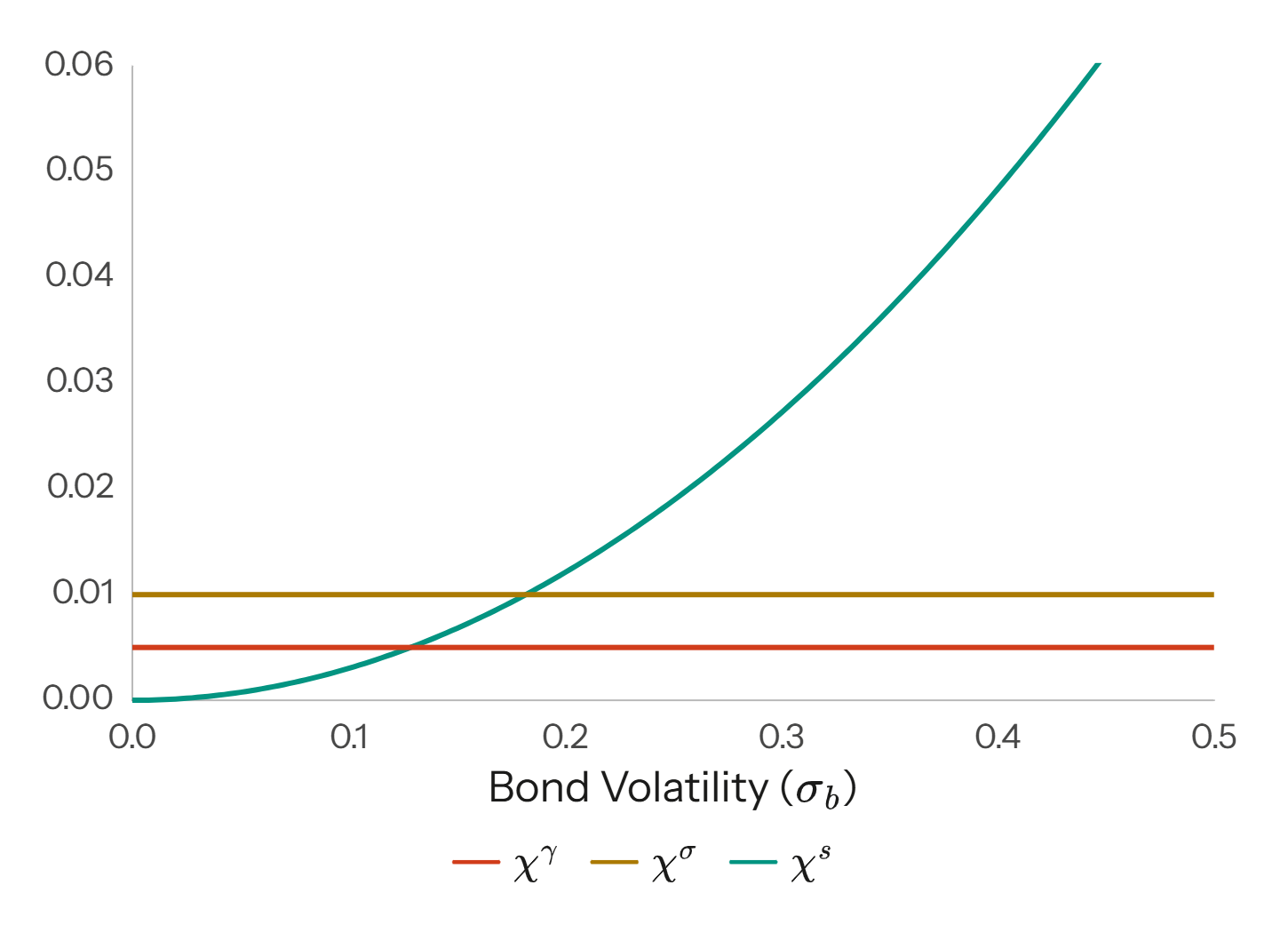}\\
Convexity&Dimensionality&Funding
\end{tabular}
\caption{Impact of the quadratic parameter $\gamma$, noise parameter $\rho$, and bond volatility $\sigma_b$ on the optimal P\&L distribution and incompleteness penalties in the mid case $\alpha\lambda=0$. The top row shows the P\&L distributions for $\gamma$ in the range $-10\le\gamma\le10$ with the $\gamma=0$ case highlighted, for $\rho$ in the range $0\le\rho\le1$ with the $\rho=0$ case highlighted, and for $\sigma_b$ in the range $0\le\sigma_b\le0.5$ with the $\sigma_b=0$ case highlighted. The bottom row demonstrates that the quadratic parameter drives the convexity penalty, the noise parameter drives the dimensionality penalty, and the bond volatility drives the funding penalty.}
\end{figure*}

The strategy for the funded derivative is assessed at the start of the interval by its entropic risk:
\begin{align}
&\frac{1}{\lambda}(\mathbb{E}[\alpha][\omega^\ast Q+\lambda(P-\delta Q-pB)]-\mathbb{E}[\alpha][\omega^\ast Q])= \\
&\qquad\eta+\frac{1}{2}\gamma(1+\rho^2)\sigma_q^2-p-\frac{1}{2}\alpha\lambda(\chi^\gamma+\chi^\sigma+\chi^s+\chi^\delta) \notag
\end{align}
where the penalties for convexity, dimensionality, funding and strategy are:
\begin{align}
\chi^\gamma&=\frac{\alpha\lambda\gamma(1+\rho^2)\sigma_q^2-\log[1+\alpha\lambda\gamma(1+\rho^2)\sigma_q^2]}{(\alpha\lambda)^2} \\
\chi^\sigma&=\frac{\epsilon^2\rho^2\sigma_q^2}{1+\alpha\lambda\gamma\rho^2\sigma_q^2} \notag \\
\chi^s&=\sigma_b^2p^2 \notag \\
\chi^\delta&=\frac{1+\alpha\lambda\gamma\rho^2\sigma_q^2}{1+\alpha\lambda\gamma(1+\rho^2)\sigma_q^2}\sigma_q^2(\delta-\delta^\ast)^2 \notag
\end{align}
and the hedge ratio that cancels the strategy penalty is:
\begin{equation}
\delta^\ast=\frac{\epsilon}{1+\alpha\lambda\gamma\rho^2\sigma_q^2}
\end{equation}
These expressions are valid for $\alpha\lambda\gamma(1+\rho^2)\sigma_q^2>-1$, which includes an open interval containing $\lambda=0$. The strategy for the funded derivative is assessed at the end of the interval by its P\&L:
\begin{align}
&P-\delta Q-pB= \\
&\qquad\eta+\frac{1}{2}\gamma(1+\rho^2)\sigma_q^2-p+(X^\gamma+X^\sigma+X^s+X^\delta) \notag
\end{align}
where the attributions for convexity, dimensionality, funding and strategy are:
\begin{align}
X^\gamma&=P-\eta-\frac{1}{2}\gamma(1+\rho^2)\sigma_q^2-\epsilon(Q+\rho\sigma_qZ) \\
X^\sigma&=\epsilon(Q+\rho\sigma_qZ)-\delta^\ast Q \notag \\
X^s&=-p(B-1) \notag \\
X^\delta&=-(\delta-\delta^\ast)Q \notag
\end{align}

Indifference dictates that the maximum entropy-ad\-just\-ed mean, attained when the strategy penalty is eliminated, matches the maximum that can be attained in the underlying market. This derives the initial derivative price:
\begin{equation}
p^\ast=\frac{2\zeta}{1+\sqrt{1+2\alpha\lambda\zeta\sigma_b^2}}
\end{equation}
where:
\begin{equation}
\zeta=\eta+\frac{1}{2}\gamma(1+\rho^2)\sigma_q^2-\frac{1}{2}\alpha\lambda(\chi^\gamma+\chi^\sigma)
\end{equation}
is the contribution to the maximum entropy-adjusted mean that does not depend on the initial price. This expression is valid for $\alpha\lambda\zeta\sigma_b^2\ge-1/2$, which includes an open interval containing $\lambda=0$.

The examples in the figures use the base parameters $\sigma_q=0.2$, $\sigma_b=0.1$, $\eta=0.5$, $\epsilon=1$, $\gamma=2$, $\rho=0.5$, $\alpha=1$, $\lambda=0$, $p=p^\ast$ and $\delta=\delta^\ast$ unless otherwise stated. At boundary parameters where the expressions are undefined, they are substituted with their finite limits. The graphs were generated using the {\em Quadratic Gaussian Entropy} application available at \href{https://mccloudresearch.com}{McCloudResearch.com}.

\section*{Gaussian processes}

The previous sections tackled the challenge of pricing and hedging over a single hedging interval when the final prices are quadratic Gaussian variables. In this section, the case of multiple hedging intervals is considered, using a Gaussian process to generate the normal factors on each interval. This framework includes many popular models, such as multifactor fractional Ornstein-Uhlenbeck processes \cite{Cheridito2003,Gatheral2018}.

The dynamic hedging schedule is specified by an increasing sequence $t_0<t_1<\cdots<t_n$ of times, where $n$ is the number of hedging intervals. On the schedule, the Gaussian process comprises the known initial vector $x_0$ and the vectors $x_1,\ldots,x_n$ that have joint normal distribution with mean and positive definite covariance:
\begin{align}
\beta_i&:=\mathbb{E}[x_i] \\
\nu_{ij}&:=\mathbb{V}[x_i,x_j] \notag
\end{align}
The process is represented in terms of independent standard normal vectors $z_1,\ldots,z_n$ using a Cholesky decomposition that respects the filtration:
\begin{equation}
x_i=\beta_i+\sum\nolimits_{k=1}^i\varsigma_{ik}z_k
\end{equation}
The coefficient matrices are generated by the scheme:
\begin{align}
\varsigma_{ij}&=(\nu_{ij}-\sum\nolimits_{k=1}^{j-1}\varsigma_{ik}\varsigma_{jk}^\dagger)\varsigma_{jj}^{-\dagger} \\
\varsigma_{ii}\varsigma_{ii}^\dagger&=\nu_{ii}-\sum\nolimits_{k=1}^{i-1}\varsigma_{ik}\varsigma_{ik}^\dagger \notag
\end{align}
iterating for $i=1,\ldots,n$ and $j=1,\ldots,i-1$, with the second line resolved by Cholesky decomposition at the $i$th interval.

In general, the Gaussian process is non-Markovian. At the $i$th interval, the factor increment $dx_i:=x_i-x_{i-1}$ conditional on the earlier factor path is normal. Furthermore, the path could provide information that partially determines the increment, and the conditional covariance can only increase as path information is removed. The increment decomposes as a sum of independent contributions \cite{Anderson2003}:
\begin{align}
&dx_i=(dx_i-\mathbb{E}[dx_i|x_1,\ldots,x_{i-1}]) \\
&\quad+\sum\nolimits_{j=2}^{i-1}(\mathbb{E}[dx_i|x_{j-1},\ldots,x_{i-1}]-\mathbb{E}[dx_i|x_j,\ldots,x_{i-1}]) \notag \\
&\quad+\mathbb{E}[dx_i|x_{i-1}] \notag
\end{align}
The first term is the error in the best prediction of the increment obtained from the entire path. The subsequent terms add to this error as factors are removed from the start of the path, and can increase the uncertainty.

The model is now framed in a way that accommodates the quadratic Gaussian analysis. This applies when the underlying prices are linear and the derivative price is quadratic in the factor increment, so make this assumption at the $i$th interval:
\begin{align}
q_i={}&\adja{\mu}_i+\adja{\sigma}_i\,dx_i \\
p_i={}&\adja{\eta}_i+\adja{\epsilon}_i\cdot dx_i+\frac{1}{2}dx_i\cdot\adja{\gamma}_i\,dx_i \notag
\end{align}
where the coefficients are known at the start of the interval. In these partial expansions, the underlying prices retain only the first order terms and the derivative price retains only the first and second order terms in the factor increment, which has mean and covariance conditional on the earlier factors:
\begin{align}
\bar{x}_i:={}&\mathbb{E}[dx_i|x_1,\ldots,x_{i-1}] \\
={}&(\beta_i-\beta_{i-1})+\sum\nolimits_{k=1}^{i-1}(\varsigma_{ik}-\varsigma_{i-1\,k})z_k \notag \\
\bar{\nu}_i:={}&\mathbb{V}[dx_i|x_1,\ldots,x_{i-1}] \notag \\
={}&\varsigma_{ii}\varsigma_{ii}^\dagger \notag
\end{align}
Provided that the fundability, independence and discounting conditions hold at the start of the interval, the previous expressions for the quadratic Gaussian model translate to this case with modifications to account for the conditional mean and covariance of the factor increment. The optimal portfolio is derived from the linear Gaussian expression for the underlying price:
\begin{equation}
\phi_i=(\adja{\sigma}_i\bar{\nu}_i\adja{\sigma}_i^\dagger)^{-1}(\adja{\mu}_i+\adja{\sigma}_i\bar{x}_i-(1+r_i)q_{i-1})
\end{equation}
where the funding rate $r_i$ enforces the self-funding condition $\phi_i\cdot q_{i-1}=0$. Define the coefficients adjusted for measure and convexity:
\begin{align}
\adjb{\sigma}_i={}&\adja{\sigma}_i(\bar{\nu}_i^{-1}+\alpha\lambda\adja{\gamma}_i)^{-1/2}\bar{\nu}_i^{-1/2} \\
\adjb{\eta}_i={}&\adja{\eta}_i-\adja{\epsilon}_i\cdot(\bar{\nu}_i\adja{\sigma}_i^\dagger\phi_i-\bar{x}_i) \notag \\
&+\frac{1}{2}(\bar{\nu}_i\adja{\sigma}_i^\dagger\phi_i-\bar{x}_i)\cdot\adja{\gamma}_i(\bar{\nu}_i\adja{\sigma}_i^\dagger\phi_i-\bar{x}_i)+\frac{1}{2}\tr[\bar{\nu}_i\adja{\gamma}_i] \notag \\
\adjb{\epsilon}_i={}&\bar{\nu}_i^{-1/2}(\bar{\nu}_i^{-1}+\alpha\lambda\adja{\gamma}_i)^{-1/2}(\adja{\epsilon}_i-\adja{\gamma}_i(\bar{\nu}_i\adja{\sigma}_i^\dagger\phi_i-\bar{x}_i)) \notag
\end{align}
The hedge portfolio is derived from the quadratic Gaussian expression for the derivative price:
\begin{align}
\delta_i^\sigma&=(\adjb{\sigma}_i\bar{\nu}_i\adjb{\sigma}_i^\dagger)^{-1}\adjb{\sigma}_i\bar{\nu}_i\adjb{\epsilon}_i \\
\delta_i^s&=(\adjb{\sigma}_i\bar{\nu}_i\adjb{\sigma}_i^\dagger)^{-1}q_{i-1} \notag \\
\delta_i^\ast&=\delta_i^\sigma+s_i\delta_i^s \notag
\end{align}
where the funding spread $s_i$ enforces the self-funding condition $\delta_i^\ast\cdot q_{i-1}=p_{i-1}$. The convexity, dimensionality and funding penalties are:
\begin{align}
\chi_i^\gamma&=\tr\!\left[\frac{\alpha\lambda\bar{\nu}_i\adja{\gamma}_i-\log[1+\alpha\lambda\bar{\nu}_i\adja{\gamma}_i]}{(\alpha\lambda)^2}\right] \\
\chi_i^\sigma&=(\adjb{\epsilon}_i-\adjb{\sigma}_i^\dagger\delta_i^\sigma)\cdot\bar{\nu}_i(\adjb{\epsilon}_i-\adjb{\sigma}_i^\dagger\delta_i^\sigma) \notag \\
\chi_i^s&=s_i^2\delta_i^s\cdot q_{i-1} \notag
\end{align}
and the indifference price for the derivative is:
\begin{align}
&p_{i-1}^\ast= \\
&\qquad\frac{1}{1+r_i}\left(\adjb{\eta}_i-\frac{1}{2}\alpha\lambda(\chi_i^\gamma+\chi_i^\sigma+{s_i^\ast}^2\delta_i^s\cdot q_{i-1})\right) \notag
\end{align}
where the indifference funding spread $s_i^\ast$ is tuned to eliminate the incremental entropy-adjusted mean arising from the introduction of the derivative.

\begin{figure*}[!t]
\begin{alignat*}{2}
&\text{Value} &\qquad\qquad &\text{Power} \\
r_i={}&\frac{q_{i-1}\cdot(\adja{\sigma}_i\bar{\nu}_i\adja{\sigma}_i^\dagger)^{-1}(\adja{\mu}_i+\adja{\sigma}_i\bar{x}_i-q_{i-1})}{q_{i-1}\cdot(\adja{\sigma}_i\bar{\nu}_i\adja{\sigma}_i^\dagger)^{-1}q_{i-1}} & &\min[1,n_x] \\
\phi_i={}&(\adja{\sigma}_i\bar{\nu}_i\adja{\sigma}_i^\dagger)^{-1}(\adja{\mu}_i+\adja{\sigma}_i\bar{x}_i-(1+r_i)q_{i-1}) & &\min[1,n_x]-n_\nu \\
\delta_i^\sigma={}&(\adjb{\sigma}_i\bar{\nu}_i\adjb{\sigma}_i^\dagger)^{-1}\adjb{\sigma}_i\bar{\nu}_i\adjb{\epsilon}_i & &0 \\
\delta_i^s={}&(\adjb{\sigma}_i\bar{\nu}_i\adjb{\sigma}_i^\dagger)^{-1}q_{i-1} & &{-}n_\nu \\
s_i={}&\frac{p_{i-1}-\delta_i^\sigma\cdot q_{i-1}}{\delta_i^s\cdot q_{i-1}} & &n_\nu \\
\chi_i^\gamma={}&\tr\!\left[\frac{\alpha\lambda\bar{\nu}_i\adja{\gamma}_i-\log[1+\alpha\lambda\bar{\nu}_i\adja{\gamma}_i]}{(\alpha\lambda)^2}\right] & &2n_\nu \\
\chi_i^\sigma={}&(\adjb{\epsilon}_i-\adjb{\sigma}_i^\dagger\delta_i^\sigma)\cdot\bar{\nu}_i(\adjb{\epsilon}_i-\adjb{\sigma}_i^\dagger\delta_i^\sigma) & &n_\nu \\
\chi_i^s={}&s_i^2\delta_i^s\cdot q_{i-1} & &n_\nu
\end{alignat*}
\caption{This table shows the funding rate and optimal portfolio, the hedging leg, funding leg and funding spread, and the convexity, dimensionality and funding penalties for the quadratic Gaussian model at the $i$th interval. In the continuous hedging limit $dt_i\to0$, these expressions are of order $dt_i$ raised to the corresponding power. The powers are expressed in terms of the powers $n_x$ for the conditional mean and $n_\nu$ for the conditional covariance of the factor increment $dx_i$.\label{fig:limits}}
\end{figure*}

Now consider the behaviour in the continuous hedging limit $dt_i:=t_i-t_{i-1}\to0$. Suppose that the coefficients of the underlying and derivative prices satisfy:
\begin{align}
\adja{\mu}_i-q_{i-1}&=O[dt_i] \\
\adja{\sigma}_i&=O[1] \notag \\
\adja{\eta}_i&=O[1] \notag \\
\adja{\epsilon}_i&=O[1] \notag \\
\adja{\gamma}_i&=O[1] \notag
\end{align}
The limits are then determined by the Gaussian process. Assume the conditional mean and covariance have limit behaviour:
\begin{align}
\bar{x}_i&=O[dt_i^{n_x}] \\
\bar{\nu}_i&=O[dt_i^{n_\nu}] \notag
\end{align}
with non-negative powers $n_x$ for the conditional mean and $n_\nu$ for the conditional covariance. To be precise, the assumption here is that the norm of the vector and all the eigenvalues of the matrix have the specified limit behaviour.

The non-zero eigenvalues of $\bar{\nu}_i\adja{\gamma}_i$ have power $n_\nu$, and this impacts the domain of validity of the quadratic Gaussian theorem for $\alpha\lambda$. When $n_\nu=0$, the domain includes a fixed open interval containing $0$. When $n_\nu>0$, the limit behaves like the small-notional limit: the finite bounds of the domain expand to $\infty$ and the domain tends to the full line. The expected change:
\begin{equation}
(\adja{\mu}_i-q_{i-1})+\adja{\sigma}_i\bar{x}_i
\end{equation}
in the underlying prices over the interval combines power $1$ and $n_x$ contributions, and has power $\min[1,n_x]$. The results from the quadratic Gaussian model then have the values and powers given in \autoref{fig:limits}. Two interesting cases are included in this framing.
\begin{description}
\item[Gaussian event]If there is a Gaussian event immediately after the start of the interval, the management of the event risk rests entirely on the initial hedge and is not improved by dynamic hedging. The limit behaviour, which in this case corresponds to the powers $n_x=0$ and $n_\nu=0$, effectively reduces to entropic risk optimisation for the Gaussian event distribution.
\item[Fractional Ornstein-Uhlenbeck]Take the Gaussian process to be a discrete fractional Ornstein-Uhlenbeck process satisfying:
\begin{equation}
dx_i=\kappa(\beta-x_{i-1})\,dt_i+\varsigma\,dw_i
\end{equation}
on a uniform schedule starting at $t_0=0$ with fixed interval $dt_i=dt$. The model is parametrised by the reversion level vector $\beta$, the reversion rate matrix $\kappa$ and the normal volatility matrix $\varsigma$. The noise $w_i$ is a vector of independent fractional Brownian motions with Hurst exponent $H$ in the range $0<H<1$, initialised at $w_0=0$, which has mean and covariance \cite{Mandelbrot1968}:
\begin{align}
\mathbb{E}[w_i]={}&0 \\
\mathbb{V}[w_i,w_j]={}&\frac{1}{2}(t_i^{2H}+t_j^{2H}-|t_i-t_j|^{2H})=f_H[i,j]\,dt^{2H} \notag
\end{align}
where:
\begin{equation}
f_H[i,j]:=\frac{1}{2}(i^{2H}+j^{2H}-|i-j|^{2H})
\end{equation}
The noise is represented in terms of the independent standard normal vectors $z_1,\ldots,z_n$ using a Cholesky decomposition that respects the filtration:
\begin{equation}
w_i=\sum\nolimits_{k=1}^i\xi_{ik}z_k\,dt^H
\end{equation}
where the scalar coefficients are generated by the scheme:
\begin{align}
\xi_{ij}&=(f_H[i,j]-\sum\nolimits_{k=1}^{j-1}\xi_{ik}\xi_{jk})/\xi_{jj} \\
\xi_{ii}^2&=f_H[i,i]-\sum\nolimits_{k=1}^{i-1}\xi_{ik}^2 \notag
\end{align}
iterating for $i=1,\ldots,n$ and $j=1,\ldots,i-1$. The mean and Cholesky coefficients of the factors are then:
\begin{align}
\beta_i&=\beta+(1-\kappa\,dt)^i(x_0-\beta) \\
\varsigma_{ik}&=(\xi_{ik}-\kappa\,dt\sum\nolimits_{l=k}^{i-1}(1-\kappa\,dt)^{i-l-1}\xi_{lk})\varsigma\,dt^H \notag
\end{align}
and the factor increment has conditional mean and covariance:
\begin{align}
\bar{x}_i&=\kappa(\beta-x_{i-1})\,dt+\varsigma\sum\nolimits_{k=1}^{i-1}(\xi_{ik}-\xi_{i-1\,k})z_k\,dt^H \\
\bar{\nu}_i&=\xi_{ii}^2\,\varsigma\varsigma^\dagger\,dt^{2H} \notag
\end{align}
Note that the diagonal coefficient is bounded for $i>1$:
\begin{equation}
\frac{\Gamma[2H+1]\sin[\pi H]}{2H\,\Gamma[H+1/2]^2}\le\xi_{ii}^2\le1-(2^{2H-1}-1)^2
\end{equation}
The lower bound follows from the moving-average representation of fractional Brownian motion and its normalisation \cite{Mandelbrot1968}, and the upper bound follows from the correlation $2^{2H-1}-1$ of consecutive increments. The coefficients in the conditional mean are also bounded, as:
\begin{equation}
\sum\nolimits_{k=1}^{i-1}(\xi_{ik}-\xi_{i-1\,k})^2=1-\xi_{ii}^2
\end{equation}
and this sum is bounded away from zero when $H\ne1/2$. The limit behaviour of the conditional mean and covariance can thus be read from the powers of $dt$ in these expressions.

With $H\ne1/2$, the process is non-Markovian and the path of the factors provides information beyond the factor at the start of the interval. From the expressions above, the powers for $i>1$ are:
\begin{align}
n_x&=H \\
n_\nu&=2H \notag
\end{align}

Entropic risk optimisation in the underlying market generates the funding rate $r_i$ and optimal portfolio $\phi_i$ with limit behaviour:
\begin{align}
r_i&=O[dt_i^H] \\
\phi_i&=O[dt_i^{-H}] \notag
\end{align}
The continuous hedging limit thus generates two undesirable features. Firstly, the funding rate per unit time has power $-(1-H)$, which means it diverges in the limit. Secondly, the maximum entropy-adjusted mean per unit time extracted by the optimal strategy has power $-1$, which means it also diverges in the limit, leading to infinite performance. This divergence in the performance of the underlying market does not occur in practice, as continuous hedging is not feasible and transaction costs negatively impact the performance.

For the derivative, the mid indifference price $p_{i-1}^\ast$ generated at $\alpha\lambda=0$ expands as:
\begin{equation}
p_{i-1}^\ast-\adja{\eta}_i=-r_i\adja{\eta}_i+\adja{\epsilon}_i\cdot(\bar{x}_i-\bar{\nu}_i\adja{\sigma}_i^\dagger\phi_i)+O[dt_i^{2H}]
\end{equation}
where the expression is $O[dt_i^H]$, and the penalties satisfy:
\begin{align}
\chi_i^\gamma&=O[dt_i^{4H}] \\
\chi_i^\sigma&=O[dt_i^{2H}] \notag \\
\chi_i^s&=O[dt_i^{2H}] \notag
\end{align}
The penalties per unit time have powers $4H-1$ for the convexity penalty and $2H-1$ for the dimensionality and funding penalties. In the limit, the convexity penalty per unit time diverges when $H<1/4$ and tends to zero when $H>1/4$, and the dimensionality and funding penalties per unit time diverge when $H<1/2$ and tend to zero when $H>1/2$. Because of these penalties, in a market with dimensionality or funding incompleteness, the entropic bid-offer spread per unit time diverges in the rough case $H<1/2$ and tends to zero in the persistent case $H>1/2$.

The case $H=1/2$ is special, as the fractional Ornstein-Uhlenbeck process is Markovian in this case. The noise is a random walk, with $f_{1/2}[i,j]=\min[i,j]$, and the coefficients are $\xi_{ik}=1$. The expressions for the conditional mean and covariance then simplify to:
\begin{align}
\bar{x}_i&=\kappa(\beta-x_{i-1})\,dt \\
\bar{\nu}_i&=\varsigma\varsigma^\dagger\,dt \notag
\end{align}
The earlier factors contribute no information beyond the factor at the start of the interval, and the powers are:
\begin{align}
n_x&=1 \\
n_\nu&=1 \notag
\end{align}

Entropic risk optimisation in the underlying market generates the funding rate $r_i$ and optimal portfolio $\phi_i$ with limit behaviour:
\begin{align}
r_i&=O[dt_i] \\
\phi_i&=O[1] \notag
\end{align}
The funding rate and maximum entropy-adjusted mean per unit time thus both have power $0$ and converge in the limit.

For the derivative, the mid indifference price $p_{i-1}^\ast$ generated at $\alpha\lambda=0$ expands as:
\begin{align}
p_{i-1}^\ast-\adja{\eta}_i={}&-r_i\adja{\eta}_i+\adja{\epsilon}_i\cdot(\bar{x}_i-\bar{\nu}_i\adja{\sigma}_i^\dagger\phi_i) \\
&+\frac{1}{2}\tr[\bar{\nu}_i\adja{\gamma}_i]+O[dt_i^2] \notag
\end{align}
where the expression is $O[dt_i]$, and the penalties satisfy:
\begin{align}
\chi_i^\gamma&=O[dt_i^2] \\
\chi_i^\sigma&=O[dt_i] \notag \\
\chi_i^s&=O[dt_i] \notag
\end{align}
The penalties per unit time have powers $1$ for the convexity penalty and $0$ for the dimensionality and funding penalties. In the limit, the convexity penalty per unit time tends to zero, and the dimensionality and funding penalties per unit time both converge. Continuous hedging eliminates the convexity penalty by using dynamic linear replication to match the convexity \cite{Boyle1980,Bertsimas2000}, while the dimensionality and funding penalties persist.
\end{description}

\section*{Conclusion}

The quadratic form assumed for the derivative price is a limitation of the quadratic Gaussian model, and its most useful application likely falls in the domain of approximations to more general models. Using the convexity profile for the derivative, the approach estimates the penalties that mitigate model risk and the attribution of these risks in realised P\&L. In this recurring exercise, the exact solutions for the single interval are applied over each hedging interval to decompose the risk and attribute the P\&L into convexity, dimensionality, funding and strategy contributions.

Extending this idea, the solutions for multiple intervals provide a way to analyse the risk mitigants accumulated over the hedging schedule. The analysis of the previous section is consistent with results in the literature for fractional processes in finance. The fundamental theorem of asset pricing of Delbaen and Schachermayer \cite{Delbaen1994} shows that a price process excludes arbitrage in the continuous hedging limit only if it is a semimartingale. With Hurst exponent $H\ne1/2$, fractional Brownian motion is not a semimartingale, and Rogers \cite{Rogers1997} and Cheridito \cite{Cheridito2003b} construct explicit arbitrages for this scenario. These arbitrages are excluded by imposing minimum hedging intervals \cite{Cheridito2003b} or incorporating transaction costs \cite{Guasoni2006}. The model is better behaved when $H=1/2$, and this recovers the semimartingale results: no arbitrage and finite incompleteness penalties per unit time even in the continuous hedging limit. Entropic pricing measures these effects in terms of the entropy-adjusted mean of the strategy and its limit behaviour.

The asymptotic analysis leaves three questions for future work. The powers say how each quantity scales but not how fine the schedule must be before the leading powers dominate, and for the fractional process the competing contributions can remain comparable over many rebalancing intervals. Convergence bounds could be derived using the explicit expressions in the quadratic Gaussian model. Transaction costs, which prevent the performance divergence in practice and exclude the arbitrages of the fractional process, have not been incorporated into the optimisation. This could extend the analysis of exponential utility hedging with transaction costs in \cite{Hodges1989,Whalley1997} to include the impact on risk penalties and the P\&L decomposition. Finally, the penalties and attributions have been determined at each interval, but they have not been accumulated over the full schedule to determine the model risk buffer and P\&L explanation over the life of a trade.

The quadratic Gaussian solution for the hedge strategy could be used in data-driven regressions, optimising the hedge in the parameter space of the model against real data, to provide effective dynamic strategies that track market evolution. In this approach, the description of the Gaussian process is rationalised by using a stochastic model with a smaller set of static parameters. For example, the discrete fractional Ornstein-Uhlenbeck process is parametrised by the Hurst exponent $H$, the reversion level vector $\beta$, the reversion rate matrix $\kappa$ and the normal volatility matrix $\varsigma$, and the hedge ratios are computed exactly from the factor path using the scheme of the previous section. This model encapsulates a rich range of dynamics. Its effectiveness as the regression form for deep hedging, including whether the penalties predict the realised hedging error, will be tested empirically in future work.

\section*{Literature review}

In a complete market the derivative is replicated exactly, as first recognised by Black and Scholes \cite{Black1973} and Merton \cite{Merton1973}, and price is a martingale in a measure equivalent to the economic measure \cite{Harrison1981}. When the market is incomplete this measure is no longer unique. To resolve the ambiguity, minimum relative entropy methods are introduced in \cite{Stutzer1996,Buchen1996,Avellaneda1997,Avellaneda1998}. Frittelli \cite{Frittelli2000} proposes the minimal entropy martingale measure for pricing in incomplete markets and links it to the maximisation of exponential utility. The measure is studied further in \cite{Grandits2002}, its duality with exponential hedging is established in \cite{Delbaen2002}, and Rouge and El Karoui \cite{Rouge2000} relate utility-based prices to entropy. Relative entropy also measures model risk \cite{Arrieta2022}: Hansen and Sargent \cite{Hansen2008} develop robust decisions under an entropy penalty on model misspecification, and Glasserman and Xu \cite{Glasserman2014} bound risk over the alternative models within a relative-entropy distance of a reference model. Cont \cite{Cont2006} quantifies the impact of model uncertainty on derivative prices.

Indifference pricing originates with Hodges and Neuberger \cite{Hodges1989}. The approach is developed under exponential utility in \cite{Becherer2003,Musiela2004,Musiela2004b,Mania2005}, including the basis risk of claims on non-traded assets \cite{Henderson2002,Davis2006} that corresponds to the dimensionality contribution of this article, with contributions collected in \cite{Carmona2009}. In parallel, coherent \cite{Artzner1999} and convex \cite{Follmer2002,Frittelli2002} risk metrics formalise risk-adjusted valuation, with conditional and dynamic versions in \cite{Detlefsen2005}, the entropic risk metric is studied in \cite{Barrieu2005,Follmer2011} and characterised as the law-invariant time-consistent risk metric in \cite{Kupper2009}, and pricing and hedging by the optimisation of convex risk metrics is developed in \cite{Xu2006,Kloppel2007,Ilhan2009}. Quadratic hedging, which minimises the variance of the hedging error \cite{Follmer1991,Duffie1991,Schweizer1995}, is the natural comparison at zero risk aversion \cite{Kallsen2011}, where the penalties of this article become variances. Funding costs enter derivative prices as valuation adjustments in \cite{Piterbarg2010,Burgard2011,McCloud2013a,McCloud2013b}; here the funding rate and spread emerge instead from the underlying market and indifference.

The entropic pricing principle applied in this article is developed in \cite{McCloud2025a,McCloud2025b}: price is derived from a log-martingale condition in the measure with minimum entropy relative to economic expectations, and the risk aversion of the investor is the parameter of the entropic risk metric. Its exact solution here rests on the algebra of quadratic Gaussian variables, with its expectation and measure groups, developed in \cite{McCloud2010,McCloud2012} for interest rate and cross-asset derivatives. The same exponential-quadratic Gaussian integrals underlie the quadratic term structure models of \cite{Ahn2002,Leippold2002} with extensions to jump-diffusion in \cite{Cheng2007}. The Girsanov relationship of the quadratic Gaussian theorem takes its name from the change of measure of Girsanov \cite{Girsanov1960}; for an exponent that is quadratic in the factors, the change of measure adjusts their covariance as well as their mean. The theorem's decomposition into normal and noncentral chi-squared variables is that of a quadratic form in normal variables, whose distribution is computed by Imhof \cite{Imhof1961} and more recently by Das \cite{Das2025}. The Ornstein-Uhlenbeck process, an example of a Gaussian process, originates in the theory of Brownian motion of Uhlenbeck and Ornstein \cite{Uhlenbeck1930} and underlies the term structure model of Vasicek \cite{Vasicek1977}, and its fractional extension \cite{Cheridito2003} underlies models of rough volatility \cite{Gatheral2018}.

Risk optimisation also underpins data-driven methods. Deep hedging \cite{Buehler2019,Buehler2021,Murray2022} utilises machine learning methods to optimise the hedging strategy, and its exponential utility indifference price coincides with the price from entropic risk optimisation.

\section*{Acknowledgements}

AI (Anthropic Claude Opus 5.5) was used to build numerical unit tests and Lean verification for the mathematical claims in this article, and the results can be viewed in the {\em Verification} application available at \href{https://mccloudresearch.com}{McCloudResearch.com}. AI also assisted with corrections, editorial review and literature search. The author takes full responsibility for the content of this article.


\begin{thebibliography}{68}%
\makeatletter
\providecommand \@ifxundefined [1]{%
 \@ifx{#1\undefined}
}%
\providecommand \@ifnum [1]{%
 \ifnum #1\expandafter \@firstoftwo
 \else \expandafter \@secondoftwo
 \fi
}%
\providecommand \@ifx [1]{%
 \ifx #1\expandafter \@firstoftwo
 \else \expandafter \@secondoftwo
 \fi
}%
\providecommand \natexlab [1]{#1}%
\providecommand \enquote  [1]{``#1''}%
\providecommand \bibnamefont  [1]{#1}%
\providecommand \bibfnamefont [1]{#1}%
\providecommand \citenamefont [1]{#1}%
\providecommand \href@noop [0]{\@secondoftwo}%
\providecommand \href [0]{\begingroup \@sanitize@url \@href}%
\providecommand \@href[1]{\@@startlink{#1}\@@href}%
\providecommand \@@href[1]{\endgroup#1\@@endlink}%
\providecommand \@sanitize@url [0]{\catcode `\\12\catcode `\$12\catcode
  `\&12\catcode `\#12\catcode `\^12\catcode `\_12\catcode `\%12\relax}%
\providecommand \@@startlink[1]{}%
\providecommand \@@endlink[0]{}%
\providecommand \url  [0]{\begingroup\@sanitize@url \@url }%
\providecommand \@url [1]{\endgroup\@href {#1}{\urlprefix }}%
\providecommand \urlprefix  [0]{URL }%
\providecommand \Eprint [0]{\href }%
\providecommand \doibase [0]{https://doi.org/}%
\providecommand \selectlanguage [0]{\@gobble}%
\providecommand \bibinfo  [0]{\@secondoftwo}%
\providecommand \bibfield  [0]{\@secondoftwo}%
\providecommand \translation [1]{[#1]}%
\providecommand \BibitemOpen [0]{}%
\providecommand \bibitemStop [0]{}%
\providecommand \bibitemNoStop [0]{.\EOS\space}%
\providecommand \EOS [0]{\spacefactor3000\relax}%
\providecommand \BibitemShut  [1]{\csname bibitem#1\endcsname}%
\let\auto@bib@innerbib\@empty
\bibitem [{\citenamefont {Boyle}\ and\ \citenamefont
  {Emanuel}(1980)}]{Boyle1980}%
  \BibitemOpen
  \bibfield  {author} {\bibinfo {author} {\bibfnamefont {P.~P.}\ \bibnamefont
  {Boyle}}\ and\ \bibinfo {author} {\bibfnamefont {D.}~\bibnamefont
  {Emanuel}},\ }\bibfield  {title} {\bibinfo {title} {Discretely adjusted
  option hedges},\ }\href {https://doi.org/10.1016/0304-405X(80)90003-3}
  {\bibfield  {journal} {\bibinfo  {journal} {Journal of Financial Economics}\
  }\textbf {\bibinfo {volume} {8}},\ \bibinfo {pages} {259} (\bibinfo {year}
  {1980})}\BibitemShut {NoStop}%
\bibitem [{\citenamefont {Bertsimas}\ \emph {et~al.}(2000)\citenamefont
  {Bertsimas}, \citenamefont {Kogan},\ and\ \citenamefont
  {Lo}}]{Bertsimas2000}%
  \BibitemOpen
  \bibfield  {author} {\bibinfo {author} {\bibfnamefont {D.}~\bibnamefont
  {Bertsimas}}, \bibinfo {author} {\bibfnamefont {L.}~\bibnamefont {Kogan}},\
  and\ \bibinfo {author} {\bibfnamefont {A.~W.}\ \bibnamefont {Lo}},\
  }\bibfield  {title} {\bibinfo {title} {When is time continuous?},\ }\href
  {https://doi.org/10.1016/S0304-405X(99)00049-5} {\bibfield  {journal}
  {\bibinfo  {journal} {Journal of Financial Economics}\ }\textbf {\bibinfo
  {volume} {55}},\ \bibinfo {pages} {173} (\bibinfo {year} {2000})}\BibitemShut
  {NoStop}%
\bibitem [{\citenamefont {Leland}(1985)}]{Leland1985}%
  \BibitemOpen
  \bibfield  {author} {\bibinfo {author} {\bibfnamefont {H.~E.}\ \bibnamefont
  {Leland}},\ }\bibfield  {title} {\bibinfo {title} {Option pricing and
  replication with transactions costs},\ }\href
  {https://doi.org/10.1111/j.1540-6261.1985.tb02383.x} {\bibfield  {journal}
  {\bibinfo  {journal} {Journal of Finance}\ }\textbf {\bibinfo {volume}
  {40}},\ \bibinfo {pages} {1283} (\bibinfo {year} {1985})}\BibitemShut
  {NoStop}%
\bibitem [{\citenamefont {Hodges}\ and\ \citenamefont
  {Neuberger}(1989)}]{Hodges1989}%
  \BibitemOpen
  \bibfield  {author} {\bibinfo {author} {\bibfnamefont {S.~D.}\ \bibnamefont
  {Hodges}}\ and\ \bibinfo {author} {\bibfnamefont {A.}~\bibnamefont
  {Neuberger}},\ }\bibfield  {title} {\bibinfo {title} {Optimal replication of
  contingent claims under transaction costs},\ }\href@noop {} {\bibfield
  {journal} {\bibinfo  {journal} {Review of Futures Markets}\ }\textbf
  {\bibinfo {volume} {8}},\ \bibinfo {pages} {222} (\bibinfo {year}
  {1989})}\BibitemShut {NoStop}%
\bibitem [{\citenamefont {Piterbarg}(2010)}]{Piterbarg2010}%
  \BibitemOpen
  \bibfield  {author} {\bibinfo {author} {\bibfnamefont {V.}~\bibnamefont
  {Piterbarg}},\ }\bibfield  {title} {\bibinfo {title} {Funding beyond
  discounting: Collateral agreements and derivatives pricing},\ }\href
  {https://www.risk.net/derivatives/1589992/funding-beyond-discounting-collateral-agreements-and-derivatives-pricing}
  {\bibfield  {journal} {\bibinfo  {journal} {Risk}\ }\textbf {\bibinfo
  {volume} {23}},\ \bibinfo {pages} {97} (\bibinfo {year} {2010})}\BibitemShut
  {NoStop}%
\bibitem [{\citenamefont {Burgard}\ and\ \citenamefont
  {Kjaer}(2011)}]{Burgard2011}%
  \BibitemOpen
  \bibfield  {author} {\bibinfo {author} {\bibfnamefont {C.}~\bibnamefont
  {Burgard}}\ and\ \bibinfo {author} {\bibfnamefont {M.}~\bibnamefont
  {Kjaer}},\ }\bibfield  {title} {\bibinfo {title} {Partial differential
  equation representations of derivatives with bilateral counterparty risk and
  funding costs},\ }\href {https://doi.org/10.21314/JCR.2011.131} {\bibfield
  {journal} {\bibinfo  {journal} {Journal of Credit Risk}\ }\textbf {\bibinfo
  {volume} {7}},\ \bibinfo {pages} {75} (\bibinfo {year} {2011})}\BibitemShut
  {NoStop}%
\bibitem [{\citenamefont {McCloud}(2013{\natexlab{a}})}]{McCloud2013a}%
  \BibitemOpen
  \bibfield  {author} {\bibinfo {author} {\bibfnamefont {P.}~\bibnamefont
  {McCloud}},\ }\bibfield  {title} {\bibinfo {title} {Collateral convexity
  complexity},\ }\href
  {https://www.risk.net/derivatives/2257630/collateral-convexity-complexity}
  {\bibfield  {journal} {\bibinfo  {journal} {Risk}\ }\textbf {\bibinfo
  {volume} {26}},\ \bibinfo {pages} {60} (\bibinfo {year}
  {2013}{\natexlab{a}})}\BibitemShut {NoStop}%
\bibitem [{\citenamefont {McCloud}(2013{\natexlab{b}})}]{McCloud2013b}%
  \BibitemOpen
  \bibfield  {author} {\bibinfo {author} {\bibfnamefont {P.}~\bibnamefont
  {McCloud}},\ }\href {https://doi.org/10.2139/ssrn.2260480} {\bibinfo {title}
  {Collateral volatility}},\ \bibinfo {howpublished} {SSRN} (\bibinfo {year}
  {2013}{\natexlab{b}})\BibitemShut {NoStop}%
\bibitem [{\citenamefont {McCloud}(2025{\natexlab{a}})}]{McCloud2025a}%
  \BibitemOpen
  \bibfield  {author} {\bibinfo {author} {\bibfnamefont {P.}~\bibnamefont
  {McCloud}},\ }\href {https://doi.org/10.48550/arXiv.2502.08613} {\bibinfo
  {title} {The relative entropy of expectation and price}} (\bibinfo {year}
  {2025}{\natexlab{a}}),\ \Eprint {https://arxiv.org/abs/2502.08613}
  {arXiv:2502.08613 [q-fin.MF]} \BibitemShut {NoStop}%
\bibitem [{\citenamefont {McCloud}(2025{\natexlab{b}})}]{McCloud2025b}%
  \BibitemOpen
  \bibfield  {author} {\bibinfo {author} {\bibfnamefont {P.}~\bibnamefont
  {McCloud}},\ }\bibfield  {title} {\bibinfo {title} {The relative entropy of
  expectation and price},\ }\href
  {https://www.risk.net/cutting-edge/7962821/the-relative-entropy-of-expectation-and-price}
  {\bibfield  {journal} {\bibinfo  {journal} {Risk}\ } (\bibinfo {year}
  {2025}{\natexlab{b}})}\BibitemShut {NoStop}%
\bibitem [{\citenamefont {Ahn}\ \emph {et~al.}(2002)\citenamefont {Ahn},
  \citenamefont {Dittmar},\ and\ \citenamefont {Gallant}}]{Ahn2002}%
  \BibitemOpen
  \bibfield  {author} {\bibinfo {author} {\bibfnamefont {D.-H.}\ \bibnamefont
  {Ahn}}, \bibinfo {author} {\bibfnamefont {R.~F.}\ \bibnamefont {Dittmar}},\
  and\ \bibinfo {author} {\bibfnamefont {A.~R.}\ \bibnamefont {Gallant}},\
  }\bibfield  {title} {\bibinfo {title} {Quadratic term structure models:
  Theory and evidence},\ }\href {https://doi.org/10.1093/rfs/15.1.243}
  {\bibfield  {journal} {\bibinfo  {journal} {Review of Financial Studies}\
  }\textbf {\bibinfo {volume} {15}},\ \bibinfo {pages} {243} (\bibinfo {year}
  {2002})}\BibitemShut {NoStop}%
\bibitem [{\citenamefont {Leippold}\ and\ \citenamefont
  {Wu}(2002)}]{Leippold2002}%
  \BibitemOpen
  \bibfield  {author} {\bibinfo {author} {\bibfnamefont {M.}~\bibnamefont
  {Leippold}}\ and\ \bibinfo {author} {\bibfnamefont {L.}~\bibnamefont {Wu}},\
  }\bibfield  {title} {\bibinfo {title} {Asset pricing under the quadratic
  class},\ }\href {https://doi.org/10.2307/3595006} {\bibfield  {journal}
  {\bibinfo  {journal} {Journal of Financial and Quantitative Analysis}\
  }\textbf {\bibinfo {volume} {37}},\ \bibinfo {pages} {271} (\bibinfo {year}
  {2002})}\BibitemShut {NoStop}%
\bibitem [{\citenamefont {McCloud}(2010)}]{McCloud2010}%
  \BibitemOpen
  \bibfield  {author} {\bibinfo {author} {\bibfnamefont {P.}~\bibnamefont
  {McCloud}},\ }\bibfield  {title} {\bibinfo {title} {Putting the smile back on
  the face of derivatives},\ }\href
  {https://www.risk.net/derivatives/1567656/putting-smile-back-face-derivatives}
  {\bibfield  {journal} {\bibinfo  {journal} {Risk}\ }\textbf {\bibinfo
  {volume} {23}} (\bibinfo {year} {2010})}\BibitemShut {NoStop}%
\bibitem [{\citenamefont {McCloud}(2012)}]{McCloud2012}%
  \BibitemOpen
  \bibfield  {author} {\bibinfo {author} {\bibfnamefont {P.}~\bibnamefont
  {McCloud}},\ }\href {https://doi.org/10.2139/ssrn.2007823} {\bibinfo {title}
  {Symmetry methods for the quadratic {Gaussian} {Libor} model}},\ \bibinfo
  {howpublished} {SSRN} (\bibinfo {year} {2012})\BibitemShut {NoStop}%
\bibitem [{\citenamefont {Cheridito}\ \emph {et~al.}(2003)\citenamefont
  {Cheridito}, \citenamefont {Kawaguchi},\ and\ \citenamefont
  {Maejima}}]{Cheridito2003}%
  \BibitemOpen
  \bibfield  {author} {\bibinfo {author} {\bibfnamefont {P.}~\bibnamefont
  {Cheridito}}, \bibinfo {author} {\bibfnamefont {H.}~\bibnamefont
  {Kawaguchi}},\ and\ \bibinfo {author} {\bibfnamefont {M.}~\bibnamefont
  {Maejima}},\ }\bibfield  {title} {\bibinfo {title} {Fractional
  {Ornstein-Uhlenbeck} processes},\ }\href {https://doi.org/10.1214/EJP.v8-125}
  {\bibfield  {journal} {\bibinfo  {journal} {Electronic Journal of
  Probability}\ }\textbf {\bibinfo {volume} {8}},\ \bibinfo {pages} {1}
  (\bibinfo {year} {2003})}\BibitemShut {NoStop}%
\bibitem [{\citenamefont {Gatheral}\ \emph {et~al.}(2018)\citenamefont
  {Gatheral}, \citenamefont {Jaisson},\ and\ \citenamefont
  {Rosenbaum}}]{Gatheral2018}%
  \BibitemOpen
  \bibfield  {author} {\bibinfo {author} {\bibfnamefont {J.}~\bibnamefont
  {Gatheral}}, \bibinfo {author} {\bibfnamefont {T.}~\bibnamefont {Jaisson}},\
  and\ \bibinfo {author} {\bibfnamefont {M.}~\bibnamefont {Rosenbaum}},\
  }\bibfield  {title} {\bibinfo {title} {Volatility is rough},\ }\href
  {https://doi.org/10.1080/14697688.2017.1393551} {\bibfield  {journal}
  {\bibinfo  {journal} {Quantitative Finance}\ }\textbf {\bibinfo {volume}
  {18}},\ \bibinfo {pages} {933} (\bibinfo {year} {2018})}\BibitemShut
  {NoStop}%
\bibitem [{\citenamefont {Buehler}\ \emph {et~al.}(2019)\citenamefont
  {Buehler}, \citenamefont {Gonon}, \citenamefont {Teichmann},\ and\
  \citenamefont {Wood}}]{Buehler2019}%
  \BibitemOpen
  \bibfield  {author} {\bibinfo {author} {\bibfnamefont {H.}~\bibnamefont
  {Buehler}}, \bibinfo {author} {\bibfnamefont {L.}~\bibnamefont {Gonon}},
  \bibinfo {author} {\bibfnamefont {J.}~\bibnamefont {Teichmann}},\ and\
  \bibinfo {author} {\bibfnamefont {B.}~\bibnamefont {Wood}},\ }\bibfield
  {title} {\bibinfo {title} {Deep hedging},\ }\href
  {https://doi.org/10.1080/14697688.2019.1571683} {\bibfield  {journal}
  {\bibinfo  {journal} {Quantitative Finance}\ }\textbf {\bibinfo {volume}
  {19}},\ \bibinfo {pages} {1271} (\bibinfo {year} {2019})}\BibitemShut
  {NoStop}%
\bibitem [{\citenamefont {Buehler}\ \emph {et~al.}(2021)\citenamefont
  {Buehler}, \citenamefont {Murray}, \citenamefont {Pakkanen},\ and\
  \citenamefont {Wood}}]{Buehler2021}%
  \BibitemOpen
  \bibfield  {author} {\bibinfo {author} {\bibfnamefont {H.}~\bibnamefont
  {Buehler}}, \bibinfo {author} {\bibfnamefont {P.}~\bibnamefont {Murray}},
  \bibinfo {author} {\bibfnamefont {M.~S.}\ \bibnamefont {Pakkanen}},\ and\
  \bibinfo {author} {\bibfnamefont {B.}~\bibnamefont {Wood}},\ }\href
  {https://doi.org/10.48550/arXiv.2111.07844} {\bibinfo {title} {Deep hedging:
  Learning to remove the drift under trading frictions with minimal equivalent
  near-martingale measures}} (\bibinfo {year} {2021}),\ \Eprint
  {https://arxiv.org/abs/2111.07844} {arXiv:2111.07844 [q-fin.CP]} \BibitemShut
  {NoStop}%
\bibitem [{\citenamefont {Murray}\ \emph {et~al.}(2022)\citenamefont {Murray},
  \citenamefont {Wood}, \citenamefont {Buehler}, \citenamefont {Wiese},\ and\
  \citenamefont {Pakkanen}}]{Murray2022}%
  \BibitemOpen
  \bibfield  {author} {\bibinfo {author} {\bibfnamefont {P.}~\bibnamefont
  {Murray}}, \bibinfo {author} {\bibfnamefont {B.}~\bibnamefont {Wood}},
  \bibinfo {author} {\bibfnamefont {H.}~\bibnamefont {Buehler}}, \bibinfo
  {author} {\bibfnamefont {M.}~\bibnamefont {Wiese}},\ and\ \bibinfo {author}
  {\bibfnamefont {M.~S.}\ \bibnamefont {Pakkanen}},\ }\bibfield  {title}
  {\bibinfo {title} {Deep hedging: Continuous reinforcement learning for
  hedging of general portfolios across multiple risk aversions},\ }in\ \href
  {https://doi.org/10.1145/3533271.3561731} {\emph {\bibinfo {booktitle}
  {Proceedings of the Third {ACM} International Conference on {AI} in
  Finance}}}\ (\bibinfo  {publisher} {Association for Computing Machinery},\
  \bibinfo {address} {New York, NY},\ \bibinfo {year} {2022})\ pp.\ \bibinfo
  {pages} {361--368}\BibitemShut {NoStop}%
\bibitem [{\citenamefont {Harrison}\ and\ \citenamefont
  {Pliska}(1981)}]{Harrison1981}%
  \BibitemOpen
  \bibfield  {author} {\bibinfo {author} {\bibfnamefont {J.~M.}\ \bibnamefont
  {Harrison}}\ and\ \bibinfo {author} {\bibfnamefont {S.~R.}\ \bibnamefont
  {Pliska}},\ }\bibfield  {title} {\bibinfo {title} {Martingales and stochastic
  integrals in the theory of continuous trading},\ }\href
  {https://doi.org/10.1016/0304-4149(81)90026-0} {\bibfield  {journal}
  {\bibinfo  {journal} {Stochastic Processes and their Applications}\ }\textbf
  {\bibinfo {volume} {11}},\ \bibinfo {pages} {215} (\bibinfo {year}
  {1981})}\BibitemShut {NoStop}%
\bibitem [{\citenamefont {Stutzer}(1996)}]{Stutzer1996}%
  \BibitemOpen
  \bibfield  {author} {\bibinfo {author} {\bibfnamefont {M.}~\bibnamefont
  {Stutzer}},\ }\bibfield  {title} {\bibinfo {title} {A simple nonparametric
  approach to derivative security valuation},\ }\href
  {https://doi.org/10.1111/j.1540-6261.1996.tb05220.x} {\bibfield  {journal}
  {\bibinfo  {journal} {Journal of Finance}\ }\textbf {\bibinfo {volume}
  {51}},\ \bibinfo {pages} {1633} (\bibinfo {year} {1996})}\BibitemShut
  {NoStop}%
\bibitem [{\citenamefont {Frittelli}(2000)}]{Frittelli2000}%
  \BibitemOpen
  \bibfield  {author} {\bibinfo {author} {\bibfnamefont {M.}~\bibnamefont
  {Frittelli}},\ }\bibfield  {title} {\bibinfo {title} {The minimal entropy
  martingale measure and the valuation problem in incomplete markets},\ }\href
  {https://doi.org/10.1111/1467-9965.00079} {\bibfield  {journal} {\bibinfo
  {journal} {Mathematical Finance}\ }\textbf {\bibinfo {volume} {10}},\
  \bibinfo {pages} {39} (\bibinfo {year} {2000})}\BibitemShut {NoStop}%
\bibitem [{\citenamefont {Rouge}\ and\ \citenamefont
  {El~Karoui}(2000)}]{Rouge2000}%
  \BibitemOpen
  \bibfield  {author} {\bibinfo {author} {\bibfnamefont {R.}~\bibnamefont
  {Rouge}}\ and\ \bibinfo {author} {\bibfnamefont {N.}~\bibnamefont
  {El~Karoui}},\ }\bibfield  {title} {\bibinfo {title} {Pricing via utility
  maximization and entropy},\ }\href {https://doi.org/10.1111/1467-9965.00093}
  {\bibfield  {journal} {\bibinfo  {journal} {Mathematical Finance}\ }\textbf
  {\bibinfo {volume} {10}},\ \bibinfo {pages} {259} (\bibinfo {year}
  {2000})}\BibitemShut {NoStop}%
\bibitem [{\citenamefont {Delbaen}\ \emph {et~al.}(2002)\citenamefont
  {Delbaen}, \citenamefont {Grandits}, \citenamefont {Rheinl{\"a}nder},
  \citenamefont {Samperi}, \citenamefont {Schweizer},\ and\ \citenamefont
  {Stricker}}]{Delbaen2002}%
  \BibitemOpen
  \bibfield  {author} {\bibinfo {author} {\bibfnamefont {F.}~\bibnamefont
  {Delbaen}}, \bibinfo {author} {\bibfnamefont {P.}~\bibnamefont {Grandits}},
  \bibinfo {author} {\bibfnamefont {T.}~\bibnamefont {Rheinl{\"a}nder}},
  \bibinfo {author} {\bibfnamefont {D.}~\bibnamefont {Samperi}}, \bibinfo
  {author} {\bibfnamefont {M.}~\bibnamefont {Schweizer}},\ and\ \bibinfo
  {author} {\bibfnamefont {C.}~\bibnamefont {Stricker}},\ }\bibfield  {title}
  {\bibinfo {title} {Exponential hedging and entropic penalties},\ }\href
  {https://doi.org/10.1111/1467-9965.02001} {\bibfield  {journal} {\bibinfo
  {journal} {Mathematical Finance}\ }\textbf {\bibinfo {volume} {12}},\
  \bibinfo {pages} {99} (\bibinfo {year} {2002})}\BibitemShut {NoStop}%
\bibitem [{\citenamefont {F{\"o}llmer}\ and\ \citenamefont
  {Schied}(2002)}]{Follmer2002}%
  \BibitemOpen
  \bibfield  {author} {\bibinfo {author} {\bibfnamefont {H.}~\bibnamefont
  {F{\"o}llmer}}\ and\ \bibinfo {author} {\bibfnamefont {A.}~\bibnamefont
  {Schied}},\ }\bibfield  {title} {\bibinfo {title} {Convex measures of risk
  and trading constraints},\ }\href {https://doi.org/10.1007/s007800200072}
  {\bibfield  {journal} {\bibinfo  {journal} {Finance and Stochastics}\
  }\textbf {\bibinfo {volume} {6}},\ \bibinfo {pages} {429} (\bibinfo {year}
  {2002})}\BibitemShut {NoStop}%
\bibitem [{\citenamefont {Frittelli}\ and\ \citenamefont
  {Rosazza~Gianin}(2002)}]{Frittelli2002}%
  \BibitemOpen
  \bibfield  {author} {\bibinfo {author} {\bibfnamefont {M.}~\bibnamefont
  {Frittelli}}\ and\ \bibinfo {author} {\bibfnamefont {E.}~\bibnamefont
  {Rosazza~Gianin}},\ }\bibfield  {title} {\bibinfo {title} {Putting order in
  risk measures},\ }\href {https://doi.org/10.1016/S0378-4266(02)00270-4}
  {\bibfield  {journal} {\bibinfo  {journal} {Journal of Banking \& Finance}\
  }\textbf {\bibinfo {volume} {26}},\ \bibinfo {pages} {1473} (\bibinfo {year}
  {2002})}\BibitemShut {NoStop}%
\bibitem [{\citenamefont {Barrieu}\ and\ \citenamefont
  {El~Karoui}(2005)}]{Barrieu2005}%
  \BibitemOpen
  \bibfield  {author} {\bibinfo {author} {\bibfnamefont {P.}~\bibnamefont
  {Barrieu}}\ and\ \bibinfo {author} {\bibfnamefont {N.}~\bibnamefont
  {El~Karoui}},\ }\bibfield  {title} {\bibinfo {title} {Inf-convolution of risk
  measures and optimal risk transfer},\ }\href
  {https://doi.org/10.1007/s00780-005-0152-0} {\bibfield  {journal} {\bibinfo
  {journal} {Finance and Stochastics}\ }\textbf {\bibinfo {volume} {9}},\
  \bibinfo {pages} {269} (\bibinfo {year} {2005})}\BibitemShut {NoStop}%
\bibitem [{\citenamefont {Musiela}\ and\ \citenamefont
  {Zariphopoulou}(2004{\natexlab{a}})}]{Musiela2004}%
  \BibitemOpen
  \bibfield  {author} {\bibinfo {author} {\bibfnamefont {M.}~\bibnamefont
  {Musiela}}\ and\ \bibinfo {author} {\bibfnamefont {T.}~\bibnamefont
  {Zariphopoulou}},\ }\bibfield  {title} {\bibinfo {title} {A valuation
  algorithm for indifference prices in incomplete markets},\ }\href
  {https://doi.org/10.1007/s00780-003-0117-0} {\bibfield  {journal} {\bibinfo
  {journal} {Finance and Stochastics}\ }\textbf {\bibinfo {volume} {8}},\
  \bibinfo {pages} {399} (\bibinfo {year} {2004}{\natexlab{a}})}\BibitemShut
  {NoStop}%
\bibitem [{\citenamefont {F{\"o}llmer}\ and\ \citenamefont
  {Schweizer}(1991)}]{Follmer1991}%
  \BibitemOpen
  \bibfield  {author} {\bibinfo {author} {\bibfnamefont {H.}~\bibnamefont
  {F{\"o}llmer}}\ and\ \bibinfo {author} {\bibfnamefont {M.}~\bibnamefont
  {Schweizer}},\ }\bibfield  {title} {\bibinfo {title} {Hedging of contingent
  claims under incomplete information},\ }in\ \href
  {https://people.math.ethz.ch/~mschweiz/Files/incinfo.pdf} {\emph {\bibinfo
  {booktitle} {Applied Stochastic Analysis}}},\ \bibinfo {series} {Stochastics
  Monographs}, Vol.~\bibinfo {volume} {5},\ \bibinfo {editor} {edited by\
  \bibinfo {editor} {\bibfnamefont {M.~H.~A.}\ \bibnamefont {Davis}}\ and\
  \bibinfo {editor} {\bibfnamefont {R.~J.}\ \bibnamefont {Elliott}}}\ (\bibinfo
   {publisher} {Gordon and Breach},\ \bibinfo {address} {London},\ \bibinfo
  {year} {1991})\ pp.\ \bibinfo {pages} {389--414}\BibitemShut {NoStop}%
\bibitem [{\citenamefont {Duffie}\ and\ \citenamefont
  {Richardson}(1991)}]{Duffie1991}%
  \BibitemOpen
  \bibfield  {author} {\bibinfo {author} {\bibfnamefont {D.}~\bibnamefont
  {Duffie}}\ and\ \bibinfo {author} {\bibfnamefont {H.~R.}\ \bibnamefont
  {Richardson}},\ }\bibfield  {title} {\bibinfo {title} {Mean-variance hedging
  in continuous time},\ }\href {https://doi.org/10.1214/aoap/1177005978}
  {\bibfield  {journal} {\bibinfo  {journal} {Annals of Applied Probability}\
  }\textbf {\bibinfo {volume} {1}},\ \bibinfo {pages} {1} (\bibinfo {year}
  {1991})}\BibitemShut {NoStop}%
\bibitem [{\citenamefont {Schweizer}(1995)}]{Schweizer1995}%
  \BibitemOpen
  \bibfield  {author} {\bibinfo {author} {\bibfnamefont {M.}~\bibnamefont
  {Schweizer}},\ }\bibfield  {title} {\bibinfo {title} {Variance-optimal
  hedging in discrete time},\ }\href {https://doi.org/10.1287/moor.20.1.1}
  {\bibfield  {journal} {\bibinfo  {journal} {Mathematics of Operations
  Research}\ }\textbf {\bibinfo {volume} {20}},\ \bibinfo {pages} {1} (\bibinfo
  {year} {1995})}\BibitemShut {NoStop}%
\bibitem [{\citenamefont {Imhof}(1961)}]{Imhof1961}%
  \BibitemOpen
  \bibfield  {author} {\bibinfo {author} {\bibfnamefont {J.~P.}\ \bibnamefont
  {Imhof}},\ }\bibfield  {title} {\bibinfo {title} {Computing the distribution
  of quadratic forms in normal variables},\ }\href
  {https://doi.org/10.1093/biomet/48.3-4.419} {\bibfield  {journal} {\bibinfo
  {journal} {Biometrika}\ }\textbf {\bibinfo {volume} {48}},\ \bibinfo {pages}
  {419} (\bibinfo {year} {1961})}\BibitemShut {NoStop}%
\bibitem [{\citenamefont {Das}(2025)}]{Das2025}%
  \BibitemOpen
  \bibfield  {author} {\bibinfo {author} {\bibfnamefont {A.}~\bibnamefont
  {Das}},\ }\bibfield  {title} {\bibinfo {title} {New methods to compute the
  generalized chi-square distribution},\ }\href
  {https://doi.org/10.1080/00949655.2025.2501401} {\bibfield  {journal}
  {\bibinfo  {journal} {Journal of Statistical Computation and Simulation}\
  }\textbf {\bibinfo {volume} {95}},\ \bibinfo {pages} {2608} (\bibinfo {year}
  {2025})}\BibitemShut {NoStop}%
\bibitem [{\citenamefont {Anderson}(2003)}]{Anderson2003}%
  \BibitemOpen
  \bibfield  {author} {\bibinfo {author} {\bibfnamefont {T.~W.}\ \bibnamefont
  {Anderson}},\ }\href
  {https://www.wiley.com/en-us/an-introduction-to-multivariate-statistical-analysis-3rd-edition-p-9780471360919}
  {\emph {\bibinfo {title} {An Introduction to Multivariate Statistical
  Analysis}}},\ \bibinfo {edition} {3rd}\ ed.,\ Wiley Series in Probability and
  Statistics\ (\bibinfo  {publisher} {Wiley},\ \bibinfo {address} {Hoboken,
  NJ},\ \bibinfo {year} {2003})\BibitemShut {NoStop}%
\bibitem [{\citenamefont {Mandelbrot}\ and\ \citenamefont
  {Van~Ness}(1968)}]{Mandelbrot1968}%
  \BibitemOpen
  \bibfield  {author} {\bibinfo {author} {\bibfnamefont {B.~B.}\ \bibnamefont
  {Mandelbrot}}\ and\ \bibinfo {author} {\bibfnamefont {J.~W.}\ \bibnamefont
  {Van~Ness}},\ }\bibfield  {title} {\bibinfo {title} {Fractional {Brownian}
  motions, fractional noises and applications},\ }\href
  {https://doi.org/10.1137/1010093} {\bibfield  {journal} {\bibinfo  {journal}
  {SIAM Review}\ }\textbf {\bibinfo {volume} {10}},\ \bibinfo {pages} {422}
  (\bibinfo {year} {1968})}\BibitemShut {NoStop}%
\bibitem [{\citenamefont {Delbaen}\ and\ \citenamefont
  {Schachermayer}(1994)}]{Delbaen1994}%
  \BibitemOpen
  \bibfield  {author} {\bibinfo {author} {\bibfnamefont {F.}~\bibnamefont
  {Delbaen}}\ and\ \bibinfo {author} {\bibfnamefont {W.}~\bibnamefont
  {Schachermayer}},\ }\bibfield  {title} {\bibinfo {title} {A general version
  of the fundamental theorem of asset pricing},\ }\href
  {https://doi.org/10.1007/BF01450498} {\bibfield  {journal} {\bibinfo
  {journal} {Mathematische Annalen}\ }\textbf {\bibinfo {volume} {300}},\
  \bibinfo {pages} {463} (\bibinfo {year} {1994})}\BibitemShut {NoStop}%
\bibitem [{\citenamefont {Rogers}(1997)}]{Rogers1997}%
  \BibitemOpen
  \bibfield  {author} {\bibinfo {author} {\bibfnamefont {L.~C.~G.}\
  \bibnamefont {Rogers}},\ }\bibfield  {title} {\bibinfo {title} {Arbitrage
  with fractional {Brownian} motion},\ }\href
  {https://doi.org/10.1111/1467-9965.00025} {\bibfield  {journal} {\bibinfo
  {journal} {Mathematical Finance}\ }\textbf {\bibinfo {volume} {7}},\ \bibinfo
  {pages} {95} (\bibinfo {year} {1997})}\BibitemShut {NoStop}%
\bibitem [{\citenamefont {Cheridito}(2003)}]{Cheridito2003b}%
  \BibitemOpen
  \bibfield  {author} {\bibinfo {author} {\bibfnamefont {P.}~\bibnamefont
  {Cheridito}},\ }\bibfield  {title} {\bibinfo {title} {Arbitrage in fractional
  {Brownian} motion models},\ }\href {https://doi.org/10.1007/s007800300101}
  {\bibfield  {journal} {\bibinfo  {journal} {Finance and Stochastics}\
  }\textbf {\bibinfo {volume} {7}},\ \bibinfo {pages} {533} (\bibinfo {year}
  {2003})}\BibitemShut {NoStop}%
\bibitem [{\citenamefont {Guasoni}(2006)}]{Guasoni2006}%
  \BibitemOpen
  \bibfield  {author} {\bibinfo {author} {\bibfnamefont {P.}~\bibnamefont
  {Guasoni}},\ }\bibfield  {title} {\bibinfo {title} {No arbitrage under
  transaction costs, with fractional {Brownian} motion and beyond},\ }\href
  {https://doi.org/10.1111/j.1467-9965.2006.00283.x} {\bibfield  {journal}
  {\bibinfo  {journal} {Mathematical Finance}\ }\textbf {\bibinfo {volume}
  {16}},\ \bibinfo {pages} {569} (\bibinfo {year} {2006})}\BibitemShut
  {NoStop}%
\bibitem [{\citenamefont {Whalley}\ and\ \citenamefont
  {Wilmott}(1997)}]{Whalley1997}%
  \BibitemOpen
  \bibfield  {author} {\bibinfo {author} {\bibfnamefont {A.~E.}\ \bibnamefont
  {Whalley}}\ and\ \bibinfo {author} {\bibfnamefont {P.}~\bibnamefont
  {Wilmott}},\ }\bibfield  {title} {\bibinfo {title} {An asymptotic analysis of
  an optimal hedging model for option pricing with transaction costs},\ }\href
  {https://doi.org/10.1111/1467-9965.00034} {\bibfield  {journal} {\bibinfo
  {journal} {Mathematical Finance}\ }\textbf {\bibinfo {volume} {7}},\ \bibinfo
  {pages} {307} (\bibinfo {year} {1997})}\BibitemShut {NoStop}%
\bibitem [{\citenamefont {Black}\ and\ \citenamefont
  {Scholes}(1973)}]{Black1973}%
  \BibitemOpen
  \bibfield  {author} {\bibinfo {author} {\bibfnamefont {F.}~\bibnamefont
  {Black}}\ and\ \bibinfo {author} {\bibfnamefont {M.}~\bibnamefont
  {Scholes}},\ }\bibfield  {title} {\bibinfo {title} {The pricing of options
  and corporate liabilities},\ }\href {https://doi.org/10.1086/260062}
  {\bibfield  {journal} {\bibinfo  {journal} {Journal of Political Economy}\
  }\textbf {\bibinfo {volume} {81}},\ \bibinfo {pages} {637} (\bibinfo {year}
  {1973})}\BibitemShut {NoStop}%
\bibitem [{\citenamefont {Merton}(1973)}]{Merton1973}%
  \BibitemOpen
  \bibfield  {author} {\bibinfo {author} {\bibfnamefont {R.~C.}\ \bibnamefont
  {Merton}},\ }\bibfield  {title} {\bibinfo {title} {Theory of rational option
  pricing},\ }\href {https://doi.org/10.2307/3003143} {\bibfield  {journal}
  {\bibinfo  {journal} {Bell Journal of Economics and Management Science}\
  }\textbf {\bibinfo {volume} {4}},\ \bibinfo {pages} {141} (\bibinfo {year}
  {1973})}\BibitemShut {NoStop}%
\bibitem [{\citenamefont {Buchen}\ and\ \citenamefont
  {Kelly}(1996)}]{Buchen1996}%
  \BibitemOpen
  \bibfield  {author} {\bibinfo {author} {\bibfnamefont {P.~W.}\ \bibnamefont
  {Buchen}}\ and\ \bibinfo {author} {\bibfnamefont {M.}~\bibnamefont {Kelly}},\
  }\bibfield  {title} {\bibinfo {title} {The maximum entropy distribution of an
  asset inferred from option prices},\ }\href {https://doi.org/10.2307/2331391}
  {\bibfield  {journal} {\bibinfo  {journal} {Journal of Financial and
  Quantitative Analysis}\ }\textbf {\bibinfo {volume} {31}},\ \bibinfo {pages}
  {143} (\bibinfo {year} {1996})}\BibitemShut {NoStop}%
\bibitem [{\citenamefont {Avellaneda}\ \emph {et~al.}(1997)\citenamefont
  {Avellaneda}, \citenamefont {Friedman}, \citenamefont {Holmes},\ and\
  \citenamefont {Samperi}}]{Avellaneda1997}%
  \BibitemOpen
  \bibfield  {author} {\bibinfo {author} {\bibfnamefont {M.}~\bibnamefont
  {Avellaneda}}, \bibinfo {author} {\bibfnamefont {C.}~\bibnamefont
  {Friedman}}, \bibinfo {author} {\bibfnamefont {R.}~\bibnamefont {Holmes}},\
  and\ \bibinfo {author} {\bibfnamefont {D.}~\bibnamefont {Samperi}},\
  }\bibfield  {title} {\bibinfo {title} {Calibrating volatility surfaces via
  relative-entropy minimization},\ }\href
  {https://doi.org/10.1080/135048697334827} {\bibfield  {journal} {\bibinfo
  {journal} {Applied Mathematical Finance}\ }\textbf {\bibinfo {volume} {4}},\
  \bibinfo {pages} {37} (\bibinfo {year} {1997})}\BibitemShut {NoStop}%
\bibitem [{\citenamefont {Avellaneda}(1998)}]{Avellaneda1998}%
  \BibitemOpen
  \bibfield  {author} {\bibinfo {author} {\bibfnamefont {M.}~\bibnamefont
  {Avellaneda}},\ }\bibfield  {title} {\bibinfo {title}
  {Minimum-relative-entropy calibration of asset-pricing models},\ }\href
  {https://doi.org/10.1142/S0219024998000242} {\bibfield  {journal} {\bibinfo
  {journal} {International Journal of Theoretical and Applied Finance}\
  }\textbf {\bibinfo {volume} {1}},\ \bibinfo {pages} {447} (\bibinfo {year}
  {1998})}\BibitemShut {NoStop}%
\bibitem [{\citenamefont {Grandits}\ and\ \citenamefont
  {Rheinl{\"a}nder}(2002)}]{Grandits2002}%
  \BibitemOpen
  \bibfield  {author} {\bibinfo {author} {\bibfnamefont {P.}~\bibnamefont
  {Grandits}}\ and\ \bibinfo {author} {\bibfnamefont {T.}~\bibnamefont
  {Rheinl{\"a}nder}},\ }\bibfield  {title} {\bibinfo {title} {On the minimal
  entropy martingale measure},\ }\href {https://doi.org/10.1214/aop/1029867119}
  {\bibfield  {journal} {\bibinfo  {journal} {Annals of Probability}\ }\textbf
  {\bibinfo {volume} {30}},\ \bibinfo {pages} {1003} (\bibinfo {year}
  {2002})}\BibitemShut {NoStop}%
\bibitem [{\citenamefont {Arrieta}(2022)}]{Arrieta2022}%
  \BibitemOpen
  \bibfield  {author} {\bibinfo {author} {\bibfnamefont {D.}~\bibnamefont
  {Arrieta}},\ }\bibfield  {title} {\bibinfo {title} {Model risk quantification
  based on relative entropy},\ }\href {https://doi.org/10.21314/JRMV.2022.020}
  {\bibfield  {journal} {\bibinfo  {journal} {Journal of Risk Model
  Validation}\ }\textbf {\bibinfo {volume} {16}},\ \bibinfo {pages} {113}
  (\bibinfo {year} {2022})}\BibitemShut {NoStop}%
\bibitem [{\citenamefont {Hansen}\ and\ \citenamefont
  {Sargent}(2008)}]{Hansen2008}%
  \BibitemOpen
  \bibfield  {author} {\bibinfo {author} {\bibfnamefont {L.~P.}\ \bibnamefont
  {Hansen}}\ and\ \bibinfo {author} {\bibfnamefont {T.~J.}\ \bibnamefont
  {Sargent}},\ }\href {https://doi.org/10.1515/9781400829385} {\emph {\bibinfo
  {title} {Robustness}}}\ (\bibinfo  {publisher} {Princeton University Press},\
  \bibinfo {address} {Princeton, NJ},\ \bibinfo {year} {2008})\BibitemShut
  {NoStop}%
\bibitem [{\citenamefont {Glasserman}\ and\ \citenamefont
  {Xu}(2014)}]{Glasserman2014}%
  \BibitemOpen
  \bibfield  {author} {\bibinfo {author} {\bibfnamefont {P.}~\bibnamefont
  {Glasserman}}\ and\ \bibinfo {author} {\bibfnamefont {X.}~\bibnamefont
  {Xu}},\ }\bibfield  {title} {\bibinfo {title} {Robust risk measurement and
  model risk},\ }\href {https://doi.org/10.1080/14697688.2013.822989}
  {\bibfield  {journal} {\bibinfo  {journal} {Quantitative Finance}\ }\textbf
  {\bibinfo {volume} {14}},\ \bibinfo {pages} {29} (\bibinfo {year}
  {2014})}\BibitemShut {NoStop}%
\bibitem [{\citenamefont {Cont}(2006)}]{Cont2006}%
  \BibitemOpen
  \bibfield  {author} {\bibinfo {author} {\bibfnamefont {R.}~\bibnamefont
  {Cont}},\ }\bibfield  {title} {\bibinfo {title} {Model uncertainty and its
  impact on the pricing of derivative instruments},\ }\href
  {https://doi.org/10.1111/j.1467-9965.2006.00281.x} {\bibfield  {journal}
  {\bibinfo  {journal} {Mathematical Finance}\ }\textbf {\bibinfo {volume}
  {16}},\ \bibinfo {pages} {519} (\bibinfo {year} {2006})}\BibitemShut
  {NoStop}%
\bibitem [{\citenamefont {Becherer}(2003)}]{Becherer2003}%
  \BibitemOpen
  \bibfield  {author} {\bibinfo {author} {\bibfnamefont {D.}~\bibnamefont
  {Becherer}},\ }\bibfield  {title} {\bibinfo {title} {Rational hedging and
  valuation of integrated risks under constant absolute risk aversion},\ }\href
  {https://doi.org/10.1016/S0167-6687(03)00140-9} {\bibfield  {journal}
  {\bibinfo  {journal} {Insurance: Mathematics and Economics}\ }\textbf
  {\bibinfo {volume} {33}},\ \bibinfo {pages} {1} (\bibinfo {year}
  {2003})}\BibitemShut {NoStop}%
\bibitem [{\citenamefont {Musiela}\ and\ \citenamefont
  {Zariphopoulou}(2004{\natexlab{b}})}]{Musiela2004b}%
  \BibitemOpen
  \bibfield  {author} {\bibinfo {author} {\bibfnamefont {M.}~\bibnamefont
  {Musiela}}\ and\ \bibinfo {author} {\bibfnamefont {T.}~\bibnamefont
  {Zariphopoulou}},\ }\bibfield  {title} {\bibinfo {title} {An example of
  indifference prices under exponential preferences},\ }\href
  {https://doi.org/10.1007/s00780-003-0112-5} {\bibfield  {journal} {\bibinfo
  {journal} {Finance and Stochastics}\ }\textbf {\bibinfo {volume} {8}},\
  \bibinfo {pages} {229} (\bibinfo {year} {2004}{\natexlab{b}})}\BibitemShut
  {NoStop}%
\bibitem [{\citenamefont {Mania}\ and\ \citenamefont
  {Schweizer}(2005)}]{Mania2005}%
  \BibitemOpen
  \bibfield  {author} {\bibinfo {author} {\bibfnamefont {M.}~\bibnamefont
  {Mania}}\ and\ \bibinfo {author} {\bibfnamefont {M.}~\bibnamefont
  {Schweizer}},\ }\bibfield  {title} {\bibinfo {title} {Dynamic exponential
  utility indifference valuation},\ }\href
  {https://doi.org/10.1214/105051605000000395} {\bibfield  {journal} {\bibinfo
  {journal} {Annals of Applied Probability}\ }\textbf {\bibinfo {volume}
  {15}},\ \bibinfo {pages} {2113} (\bibinfo {year} {2005})}\BibitemShut
  {NoStop}%
\bibitem [{\citenamefont {Henderson}(2002)}]{Henderson2002}%
  \BibitemOpen
  \bibfield  {author} {\bibinfo {author} {\bibfnamefont {V.}~\bibnamefont
  {Henderson}},\ }\bibfield  {title} {\bibinfo {title} {Valuation of claims on
  nontraded assets using utility maximization},\ }\href
  {https://doi.org/10.1111/j.1467-9965.2002.tb00129.x} {\bibfield  {journal}
  {\bibinfo  {journal} {Mathematical Finance}\ }\textbf {\bibinfo {volume}
  {12}},\ \bibinfo {pages} {351} (\bibinfo {year} {2002})}\BibitemShut
  {NoStop}%
\bibitem [{\citenamefont {Davis}(2006)}]{Davis2006}%
  \BibitemOpen
  \bibfield  {author} {\bibinfo {author} {\bibfnamefont {M.~H.~A.}\
  \bibnamefont {Davis}},\ }\bibfield  {title} {\bibinfo {title} {Optimal
  hedging with basis risk},\ }in\ \href
  {https://doi.org/10.1007/978-3-540-30788-4_8} {\emph {\bibinfo {booktitle}
  {From Stochastic Calculus to Mathematical Finance: The {Shiryaev}
  Festschrift}}},\ \bibinfo {editor} {edited by\ \bibinfo {editor}
  {\bibfnamefont {Y.}~\bibnamefont {Kabanov}}, \bibinfo {editor} {\bibfnamefont
  {R.}~\bibnamefont {Liptser}},\ and\ \bibinfo {editor} {\bibfnamefont
  {J.}~\bibnamefont {Stoyanov}}}\ (\bibinfo  {publisher} {Springer},\ \bibinfo
  {address} {Berlin, Heidelberg},\ \bibinfo {year} {2006})\ pp.\ \bibinfo
  {pages} {169--187}\BibitemShut {NoStop}%
\bibitem [{\citenamefont {Carmona}(2009)}]{Carmona2009}%
  \BibitemOpen
  \bibinfo {editor} {\bibfnamefont {R.}~\bibnamefont {Carmona}},\ ed.,\ \href
  {https://doi.org/10.1515/9781400833115} {\emph {\bibinfo {title}
  {Indifference Pricing: Theory and Applications}}},\ Princeton Series in
  Financial Engineering\ (\bibinfo  {publisher} {Princeton University Press},\
  \bibinfo {address} {Princeton, NJ},\ \bibinfo {year} {2009})\BibitemShut
  {NoStop}%
\bibitem [{\citenamefont {Artzner}\ \emph {et~al.}(1999)\citenamefont
  {Artzner}, \citenamefont {Delbaen}, \citenamefont {Eber},\ and\ \citenamefont
  {Heath}}]{Artzner1999}%
  \BibitemOpen
  \bibfield  {author} {\bibinfo {author} {\bibfnamefont {P.}~\bibnamefont
  {Artzner}}, \bibinfo {author} {\bibfnamefont {F.}~\bibnamefont {Delbaen}},
  \bibinfo {author} {\bibfnamefont {J.-M.}\ \bibnamefont {Eber}},\ and\
  \bibinfo {author} {\bibfnamefont {D.}~\bibnamefont {Heath}},\ }\bibfield
  {title} {\bibinfo {title} {Coherent measures of risk},\ }\href
  {https://doi.org/10.1111/1467-9965.00068} {\bibfield  {journal} {\bibinfo
  {journal} {Mathematical Finance}\ }\textbf {\bibinfo {volume} {9}},\ \bibinfo
  {pages} {203} (\bibinfo {year} {1999})}\BibitemShut {NoStop}%
\bibitem [{\citenamefont {Detlefsen}\ and\ \citenamefont
  {Scandolo}(2005)}]{Detlefsen2005}%
  \BibitemOpen
  \bibfield  {author} {\bibinfo {author} {\bibfnamefont {K.}~\bibnamefont
  {Detlefsen}}\ and\ \bibinfo {author} {\bibfnamefont {G.}~\bibnamefont
  {Scandolo}},\ }\bibfield  {title} {\bibinfo {title} {Conditional and dynamic
  convex risk measures},\ }\href {https://doi.org/10.1007/s00780-005-0159-6}
  {\bibfield  {journal} {\bibinfo  {journal} {Finance and Stochastics}\
  }\textbf {\bibinfo {volume} {9}},\ \bibinfo {pages} {539} (\bibinfo {year}
  {2005})}\BibitemShut {NoStop}%
\bibitem [{\citenamefont {F{\"o}llmer}\ and\ \citenamefont
  {Knispel}(2011)}]{Follmer2011}%
  \BibitemOpen
  \bibfield  {author} {\bibinfo {author} {\bibfnamefont {H.}~\bibnamefont
  {F{\"o}llmer}}\ and\ \bibinfo {author} {\bibfnamefont {T.}~\bibnamefont
  {Knispel}},\ }\bibfield  {title} {\bibinfo {title} {Entropic risk measures:
  Coherence vs. convexity, model ambiguity and robust large deviations},\
  }\href {https://doi.org/10.1142/S0219493711003334} {\bibfield  {journal}
  {\bibinfo  {journal} {Stochastics and Dynamics}\ }\textbf {\bibinfo {volume}
  {11}},\ \bibinfo {pages} {333} (\bibinfo {year} {2011})}\BibitemShut
  {NoStop}%
\bibitem [{\citenamefont {Kupper}\ and\ \citenamefont
  {Schachermayer}(2009)}]{Kupper2009}%
  \BibitemOpen
  \bibfield  {author} {\bibinfo {author} {\bibfnamefont {M.}~\bibnamefont
  {Kupper}}\ and\ \bibinfo {author} {\bibfnamefont {W.}~\bibnamefont
  {Schachermayer}},\ }\bibfield  {title} {\bibinfo {title} {Representation
  results for law invariant time consistent functions},\ }\href
  {https://doi.org/10.1007/s11579-009-0019-9} {\bibfield  {journal} {\bibinfo
  {journal} {Mathematics and Financial Economics}\ }\textbf {\bibinfo {volume}
  {2}},\ \bibinfo {pages} {189} (\bibinfo {year} {2009})}\BibitemShut {NoStop}%
\bibitem [{\citenamefont {Xu}(2006)}]{Xu2006}%
  \BibitemOpen
  \bibfield  {author} {\bibinfo {author} {\bibfnamefont {M.}~\bibnamefont
  {Xu}},\ }\bibfield  {title} {\bibinfo {title} {Risk measure pricing and
  hedging in incomplete markets},\ }\href
  {https://doi.org/10.1007/s10436-005-0023-x} {\bibfield  {journal} {\bibinfo
  {journal} {Annals of Finance}\ }\textbf {\bibinfo {volume} {2}},\ \bibinfo
  {pages} {51} (\bibinfo {year} {2006})}\BibitemShut {NoStop}%
\bibitem [{\citenamefont {Kl{\"o}ppel}\ and\ \citenamefont
  {Schweizer}(2007)}]{Kloppel2007}%
  \BibitemOpen
  \bibfield  {author} {\bibinfo {author} {\bibfnamefont {S.}~\bibnamefont
  {Kl{\"o}ppel}}\ and\ \bibinfo {author} {\bibfnamefont {M.}~\bibnamefont
  {Schweizer}},\ }\bibfield  {title} {\bibinfo {title} {Dynamic indifference
  valuation via convex risk measures},\ }\href
  {https://doi.org/10.1111/j.1467-9965.2007.00317.x} {\bibfield  {journal}
  {\bibinfo  {journal} {Mathematical Finance}\ }\textbf {\bibinfo {volume}
  {17}},\ \bibinfo {pages} {599} (\bibinfo {year} {2007})}\BibitemShut
  {NoStop}%
\bibitem [{\citenamefont {{\.I}lhan}\ \emph {et~al.}(2009)\citenamefont
  {{\.I}lhan}, \citenamefont {Jonsson},\ and\ \citenamefont
  {Sircar}}]{Ilhan2009}%
  \BibitemOpen
  \bibfield  {author} {\bibinfo {author} {\bibfnamefont {A.}~\bibnamefont
  {{\.I}lhan}}, \bibinfo {author} {\bibfnamefont {M.}~\bibnamefont {Jonsson}},\
  and\ \bibinfo {author} {\bibfnamefont {R.}~\bibnamefont {Sircar}},\
  }\bibfield  {title} {\bibinfo {title} {Optimal static-dynamic hedges for
  exotic options under convex risk measures},\ }\href
  {https://doi.org/10.1016/j.spa.2009.06.009} {\bibfield  {journal} {\bibinfo
  {journal} {Stochastic Processes and their Applications}\ }\textbf {\bibinfo
  {volume} {119}},\ \bibinfo {pages} {3608} (\bibinfo {year}
  {2009})}\BibitemShut {NoStop}%
\bibitem [{\citenamefont {Kallsen}\ and\ \citenamefont
  {Rheinl{\"a}nder}(2011)}]{Kallsen2011}%
  \BibitemOpen
  \bibfield  {author} {\bibinfo {author} {\bibfnamefont {J.}~\bibnamefont
  {Kallsen}}\ and\ \bibinfo {author} {\bibfnamefont {T.}~\bibnamefont
  {Rheinl{\"a}nder}},\ }\bibfield  {title} {\bibinfo {title} {Asymptotic
  utility-based pricing and hedging for exponential utility},\ }\href
  {https://doi.org/10.1524/stnd.2011.1027} {\bibfield  {journal} {\bibinfo
  {journal} {Statistics \& Decisions}\ }\textbf {\bibinfo {volume} {28}},\
  \bibinfo {pages} {17} (\bibinfo {year} {2011})}\BibitemShut {NoStop}%
\bibitem [{\citenamefont {Cheng}\ and\ \citenamefont
  {Scaillet}(2007)}]{Cheng2007}%
  \BibitemOpen
  \bibfield  {author} {\bibinfo {author} {\bibfnamefont {P.}~\bibnamefont
  {Cheng}}\ and\ \bibinfo {author} {\bibfnamefont {O.}~\bibnamefont
  {Scaillet}},\ }\bibfield  {title} {\bibinfo {title} {Linear-quadratic
  jump-diffusion modeling},\ }\href
  {https://doi.org/10.1111/j.1467-9965.2007.00316.x} {\bibfield  {journal}
  {\bibinfo  {journal} {Mathematical Finance}\ }\textbf {\bibinfo {volume}
  {17}},\ \bibinfo {pages} {575} (\bibinfo {year} {2007})}\BibitemShut
  {NoStop}%
\bibitem [{\citenamefont {Girsanov}(1960)}]{Girsanov1960}%
  \BibitemOpen
  \bibfield  {author} {\bibinfo {author} {\bibfnamefont {I.~V.}\ \bibnamefont
  {Girsanov}},\ }\bibfield  {title} {\bibinfo {title} {On transforming a
  certain class of stochastic processes by absolutely continuous substitution
  of measures},\ }\href {https://doi.org/10.1137/1105027} {\bibfield  {journal}
  {\bibinfo  {journal} {Theory of Probability and Its Applications}\ }\textbf
  {\bibinfo {volume} {5}},\ \bibinfo {pages} {285} (\bibinfo {year}
  {1960})}\BibitemShut {NoStop}%
\bibitem [{\citenamefont {Uhlenbeck}\ and\ \citenamefont
  {Ornstein}(1930)}]{Uhlenbeck1930}%
  \BibitemOpen
  \bibfield  {author} {\bibinfo {author} {\bibfnamefont {G.~E.}\ \bibnamefont
  {Uhlenbeck}}\ and\ \bibinfo {author} {\bibfnamefont {L.~S.}\ \bibnamefont
  {Ornstein}},\ }\bibfield  {title} {\bibinfo {title} {On the theory of the
  {Brownian} motion},\ }\href {https://doi.org/10.1103/PhysRev.36.823}
  {\bibfield  {journal} {\bibinfo  {journal} {Physical Review}\ }\textbf
  {\bibinfo {volume} {36}},\ \bibinfo {pages} {823} (\bibinfo {year}
  {1930})}\BibitemShut {NoStop}%
\bibitem [{\citenamefont {Vasicek}(1977)}]{Vasicek1977}%
  \BibitemOpen
  \bibfield  {author} {\bibinfo {author} {\bibfnamefont {O.}~\bibnamefont
  {Vasicek}},\ }\bibfield  {title} {\bibinfo {title} {An equilibrium
  characterization of the term structure},\ }\href
  {https://doi.org/10.1016/0304-405X(77)90016-2} {\bibfield  {journal}
  {\bibinfo  {journal} {Journal of Financial Economics}\ }\textbf {\bibinfo
  {volume} {5}},\ \bibinfo {pages} {177} (\bibinfo {year} {1977})}\BibitemShut
  {NoStop}%
\end{thebibliography}
\end{document}